\documentclass[lettersize,journal]{IEEEtran}
\usepackage{amsmath,amsfonts,amssymb}
\usepackage{bm}
\usepackage{graphicx}
\usepackage{amsthm}
\usepackage{subfigure}
\usepackage{textcomp}
\usepackage{threeparttable}
\usepackage{makecell}
\usepackage{multirow,multicol}
\usepackage[colorlinks]{hyperref}
\theoremstyle{remark}
\newtheorem{theorem}{Theorem}
\newtheorem{lemma}{Lemma}
\newtheorem{corollary}{Corollary}

\newcommand{\var}{\operatorname{var}}
\newcommand{\cov}{\operatorname{cov}}

\begin{document}
\title{Phased Ultra Massive Array (PUMA)}

\author{Hanjiang Hong,~\IEEEmembership{Member,~IEEE}, 
Kai-Kit Wong,~\IEEEmembership{Fellow,~IEEE}, 
Xusheng Zhu,~\IEEEmembership{Member,~IEEE},\\
Chenguang Rao,
Dazhi He,~\IEEEmembership{Senior Member,~IEEE}, and 
Hyundong~Shin,~\IEEEmembership{Fellow,~IEEE}
\vspace{-5mm}

\thanks{The work of K. K. Wong is supported by the Engineering and Physical Sciences Research Council (EPSRC) under Grant EP/W026813/1.}
\thanks{The work of X. Zhu was supported by the NSFC under Grant 624B2094 and the Outstanding Doctoral Graduates Development Scholarship of Shanghai Jiao Tong University.}
\thanks{The work of D. He is supported by the National Natural Science Foundation of China (62271316).}
\thanks{The work of H. Shin is supported by the National Research Foundation of Korea (NRF) grant funded by the Korean government (RS-2025-00556064 and RS-2025-25442355), and by the Ministry of Science and ICT (MSIT), Korea, under the ITRC (Information Technology Research Center) support program (IITP-2025-RS-2021-II212046), supervised by the IITP (Institute for Information \& Communications Technology Planning \& Evaluation).}
\thanks{H. Hong, K. K. Wong, X. Zhu, and C. Rao are with the Department of Electronic and Electrical Engineering, University College London, London, United Kingdom. K. K. Wong is also affiliated with the Department of Electronic Engineering, Kyung Hee University, Yongin-si, Gyeonggi-do 17104, Korea (e-mail: \{hanjiang.hong, kai-kit.wong, xusheng.zhu, chenguang.rao\}@ucl.ac.uk).}
\thanks{D. He is with the Cooperative Medianet Innovation Center (CMIC), School of Information and Electronic Engineering, Shanghai Jiao Tong University, Shanghai 200240, China (e-mail:hedazhi@sjtu.edu.cn).}
\thanks{H. Shin is with the Department of Electronics and Information Convergence Engineering, Kyung Hee University, Yongin-si, Gyeonggi-do 17104, Republic of Korea (e-mail: hshin@khu.ac.kr).}

\thanks{The reproducible MATLAB code is online at \url{https://github.com/honghjvvv/Phased-Ultra-Massive-Array-PUMA.git}.}

\thanks{Corresponding author: Kai-Kit Wong.}
}
\maketitle

\begin{abstract}
This paper proposes a novel multiple-access framework, termed the phased ultra massive array (PUMA), which exploits the distinctive spatial flexibility of fluid antenna systems (FAS) at the user equipment (UE). Building upon fluid antenna multiple access (FAMA) and compact ultra-massive antenna array (CUMA), PUMA incorporates a phased array for signal aggregation. This architecture enables the UE to inherently mitigate co-user interference within the spatial domain without necessitating complex interference cancellation. A primary advantage of PUMA lies in its hardware efficiency: by implementing phase shifting and signal combining in the analog domain, it achieves high antenna gain while requiring only a minimal number of radio-frequency (RF) chains, potentially a single RF chain. Comprehensive theoretical analysis of the achievable data rate is provided, complemented by extensive simulations that validate the framework. The results demonstrate that PUMA markedly outperforms FAMA and CUMA architectures, particularly for UEs with a single RF chain, offering a robust and scalable solution for interference-insensitive massive connectivity in sixth-generation (6G) systems. 
\end{abstract}


\section{Introduction}
\IEEEPARstart{T}{he transition} to the sixth-generation (6G) wireless networks is driven by the increasing demand for ubiquitous, high-data-rate, and massive connectivity \cite{andrews20246gtakes,ngo2024ultradense}. As the internet of things (IoT) advances into internet of everything (IoE), 6G is anticipated to support connection densities reaching up to $10^7$ devices/km$^2$ \cite{ding2025iot,nguyen20226ginternet,silva2025distributed}. This unprecedented scale of connectivity poses a substantial challenge to the underlying physical-layer technologies, requiring a fundamental change in the management of multiuser interference and the utilization of the available electromagnetic spectrum.

To address these requirements, advanced technologies such as multiuser multiple-input multiple-output (MIMO) \cite{wang2024extremely} and non-orthogonal multiple access (NOMA) \cite{clercks2024multiple,ahmed2024unveil} have been the subject of extensive research. Massive MIMO employs large-scale antenna arrays at the base station (BS) to attain high spectral efficiency through spatial multiplexing. Nevertheless, its performance is substantially dependent on the precise acquisition of channel state information (CSI), which becomes increasingly challenging in high-mobility or large-user/antenna count cases due to pilot contamination and signaling overhead. Conversely, NOMA permits multiple users to share the same resource block, albeit typically necessitating successive interference cancellation (SIC) at the user equipment (UE) to handle the inter-user interference. The complexity of SIC scales exponentially with the number of users, and its effectiveness is highly sensitive to power-level disparities and error propagation, which prohibits its adoption in 6G. 

Recently, fluid antenna multiple access (FAMA) was introduced as a radical approach for massive connectivity \cite{wong2022FAMA,hong2026famasurvey}. FAMA eliminates the need for precoding and SIC, instead depending on the repositioning capability of the antenna enabled by the fluid antenna system (FAS) concept \cite{wong2020FAS,wong2021FAS}. FAS signifies a shift toward a flexible radio architecture in which the antenna element can be reconfigured to provide diversity and capacity benefits \cite{New2024aTutorial,hong2025contemporary,New2026jsac}. In practice, FAS can take many forms, including liquid-based antennas \cite{shen2024design,Shamim-2025}, metamaterials \cite{Liu2025programmable,Zhangjsac2026}, and pixel-based antennas \cite{zhang2024pixel,liu2025iot}. Numerous studies on FAS have already been conducted, covering areas such as physical-layer performance \cite{espinosa2024anew,zhu2025fluid,Khammassi2023,zhang2025finite}, CSI estimation \cite{xu2023channel,zhang2025successive}, integration with reconfigurable intelligent surfaces (RIS) \cite{ghadi2024on,xiao2025fluid}, sensing \cite{zhou2024sensing,zhang2026jsac}, and full-duplexing \cite{hong2026FDSIC,tang2026FD}. More recently, positive findings have been reported regarding its integration into fifth-generation (5G) New Radio (NR) systems \cite{hong2025fasofdm}.

The reconfigurability of FAS allows the receiver to leverage the fine-grained variations of the fading envelope in the spatial domain, effectively exploiting deep interference fades or the peaks of the desired signal even within a space-constrained UE. In FAMA, inter-user interference is inherently mitigated at the receiver by utilizing the high spatial resolution provided by a sufficiently large array of ports within the FAS. This enables multiple users to transmit simultaneously on the same time-frequency resource block without requiring complex coordination from the BS and/or intricate SIC at the UE. Various forms of FAMA schemes have been developed, including fast FAMA (\emph{f}FAMA) \cite{wong2022fast}, slow FAMA (\emph{s}FAMA) \cite{wong2023sFAMA,Xu2024revisiting,coma2024slow,zhang2025sFAMA,hong2025multiport}, and coded FAMA \cite{hong20245gcoded,hong2025coded,waqar2025turbocharging,hong2025Downlink}. To further enhance this concept for ultra-dense deployments, the compact ultra-massive antenna array (CUMA) was introduced, emphasizing the activation of numerous ports and the aggregation of received signals in the analog domain to generate the signal for digital communication \cite{Wong2024cuma}. Typically, CUMA considers configurations with only one or two radio frequency (RF) chains at the receiver, to keep the hardware at low cost. It was later expanded to accommodate four \cite{wong2024compact} or any even number of RF chains \cite[Section V]{New2024aTutorial}. Recent research has focused on optimizing port selection within the CUMA architecture \cite{rao2025geometric}. While FAMA and CUMA effectively eliminate the need for BS-side {\color{black}CSI-based precoding} and UE-side SIC, their performance can be indeed improved if co-phasing can be performed in the analog domain before signal aggregation of the selected ports.

In fact, the feasibility of deploying large-scale phase shifter arrays has improved due to recent advancements in complementary metal-oxide-semiconductor (CMOS) technology and integrated circuit (IC) fabrication \cite{uchino2025PS}. Analog-phase shifters offer substantial reductions in power consumption and occupy smaller chip areas \cite{morishita2025_150}. Motivated by this, this paper proposes a novel FAMA framework termed the phased ultra massive array (PUMA), an advancement of the FAMA and CUMA architectures, designed to achieve superior interference mitigation performance. From an antenna design perspective, PUMA integrates FAS at each UE with an analog-domain phased array for aggregation. By selectively activating or deactivating certain antenna ports and integrating the phased array, PUMA markedly improves received signal strength while preserving the natural interference-agnostic features of its predecessors. In PUMA, FAS ports activation, phase shifters, and signal combination are realized entirely in the analog domain, enabling the exploitation of ultra-massive antenna gains with a minimal number of RF chains, enabled by the FAS concept.

{\color{black}Although PUMA employs analog phase shifters, its operating principle differs from conventional hybrid beamforming (HBF) \cite{sohrabi2016HBF,gao2024hbf}. Conventional HBF operates over a fixed antenna array and typically relies on CSI-dependent joint analog and digital beamforming to form directional beams and manage multiuser interference. In contrast, PUMA additionally exploits the reconfigurable spatial domain of FAS, where the set of active ports itself provides a spatial degree of freedom. The phase shifters are used primarily to align the desired channel components across the selected ports, whereas co-user interference is mitigated through the spatial fluctuations offered by FAS and the random interference incorporation rather than explicit beamforming or sophisticated interference cancellation. Consequently, PUMA preserves the FAMA/CUMA philosophy of avoiding CSI-based precoding at the BS and complex multiuser interference processing at the UE, while converting the random analog aggregation of CUMA into coherent desired-signal aggregation.}

The goal of this work is to illustrate that by transferring the interference management from the digital baseband to the spatial/analog domain through the use of FAS and phased array technology, we can achieve robust multiuser communication without the complexity of {\color{black}CSI-based centralized }precoding at the BS and SIC at the UE. The main contributions are summarized as follows:
\begin{itemize}
\item We propose the PUMA framework, {\color{black}which jointly exploits FAS port activation and analog phase control while requiring only one or a limited number of RF chains at the UE. The phase shifters facilitate the constructive alignment of signals, thereby converting spatial signal aggregation from random CUMA aggregation to coherent desired aggregation and} enhancing the array gain compared to conventional FAMA. 
\item We present a theoretical analysis of the data-rate performance of PUMA {\color{black}under rich scattering}, deriving closed-form expressions for the probability density function (PDF) of the signal-to-interference ratio (SIR). Utilizing the PDF, we then derive the achievable rate in the presence of substantial co-user interference. The interference resilience of the UE is expected to improve as the FAS configuration increases.
\item Extensive simulation results are presented to evaluate PUMA's performance {\color{black}under both rich and finite scattering models}. Our results demonstrate that the network rate of PUMA increases with the number of UEs and greatly outperforms CUMA and \emph{s}FAMA, especially when constrained to only a single RF chain.
\end{itemize}

The remainder of this paper is structured as follows. Section~\ref{sec:NetMod} delineates the system model for the conventional FAMA and CUMA frameworks. Then Section~\ref{sec:PUMA} details the signal model and the receiver architecture of the proposed PUMA system. Subsequently, Section~\ref{sec:Perf} offers the SIR and rate analysis for PUMA. Numerical results are provided and discussed in Section~\ref{Sec:Sim}, and the paper concludes with Section~\ref{sec:conclusion}.

{\em Notations:} Throughout, scalars are represented by lowercase letters while vectors and matrices are denoted by lowercase and uppercase boldface letters, respectively. Transpose and hermitian operations are denoted by superscript $T$ and $\dag$, respectively. For a complex scalar $x$, $\lvert x \rvert$ represents its modulus. For a set $\mathcal{X}$, $\lvert \mathcal{X} \rvert$ represents its cardinal number.

\section{FAMA and CUMA Network Model}\label{sec:NetMod}
{\color{black}In this paper, we consider a downlink network model where a BS with multiple conventional fixed position antennas (FPAs) transmits to $U$ UEs, without using CSI-based precoding. It is assumed that each transmit reception point (TRP) of the BS serves one UE.} Each UE is equipped with a two-dimensional FAS (2D-FAS) consisting of $N=(N_1 \times N_2)$ ports within a physical size of $W_1\lambda \times W_2 \lambda$, where $\lambda$ denotes the carrier wavelength. Within the 2D-FAS physical space, $N_i$ ports are evenly distributed along a linear dimension of length $W_i \lambda$ for $i \in \{1,2\}$. For simplicity, we map the antenna port $(k_1, k_2) \to k: k = (k_1-1)\times N_2 +k_2$, where $k_1 \in \{1, \dots, N_1\}$, $k_2 \in \{1, \dots, N_2\}$, and $k \in \{1, \dots, N\}$. The ports can be closely spaced, so their channels are correlated.

Hypothetically, if all the ports are activated, the received signals at the $u$-th UE in vector form can be given by
\begin{equation}
\bm{r}^{(u)} 
= \bm{g}^{(u,u)} s_{u} + \sum_{\substack{\tilde{u} = 1\\ \tilde{u}\neq u}}^U \bm{g}^{(\tilde{u},u)} s_{\tilde{u}} + \bm{\eta}^{(u)},
\end{equation}
{\color{black} in which $\bm{g}^{(\tilde{u},u)}  = [{g}^{(\tilde{u},u)}_1, \dots, {g}^{(\tilde{u},u)}_N]^T$ denotes the channel vector from the $\tilde{u}$-th transmit antenna to UE $u$,} $s_{u}$ is the transmitted symbol for UE $u$ satisfying $\mathbb{E} [|s_{u}|^2] = \sigma_s^2$, and $\bm{\eta}^{(u)} = [\eta^{(u)}_1,\dots,\eta^{(u)}_N]^T$ represents the zero-mean complex Gaussian noise vector at the ports of UE $u$ whose elements are independent, identically distributed (i.i.d.), with variance of $\sigma_\eta^2$. {\color{black}The received signal at the $k$-th port of the $u$-th UE can also be denoted as $r_k^{(u)}\triangleq [\bm{r}^{(u)}]_k$.} 

Regarding the channel vector $\bm{g}^{(\tilde{u},u)}$, a block fading finite-scattering channel model is employed, expressed as \cite{buzzi2016clustered}
\begin{multline}
\boldsymbol{g}^{(\tilde{u},u)} =\sqrt{\frac{K\sigma_{(\tilde{u},u)}^2}{K+1}} e^{j\delta_{(\tilde{u},u)}} \boldsymbol{a}\left(\theta_0^{(\tilde{u},u)}, \phi_0^{(\tilde{u},u)}\right)\\ 
+ \sqrt{\frac{\sigma_{(\tilde{u},u)}^2}{N_p(K+1)}} \sum_{l=1}^{N_p} \alpha_l^{(\tilde{u},u)} \boldsymbol{a}\left(\theta_l^{(\tilde{u},u)}, \phi_l^{(\tilde{u},u)}\right),
\end{multline}
where $K$ denotes the Rice factor, $\sigma_{(\tilde{u},u)}^2$ represents the channel power, $\delta_{(\tilde{u},u)}$ signifies the phase of the line-of-sight (LoS) component, $N_p$ indicates the number of scattered components, $\alpha_l^{(\tilde{u},u)}$ is the random complex coefficient of the $l$-th scattered path, and $\boldsymbol{a}(\theta_l^{(\tilde{u},u)}, \phi_l^{(\tilde{u},u)})$ is the steering vector, with $\theta_l^{(\tilde{u},u)}$ and $\phi_l^{(\tilde{u},u)}$ representing the azimuth and elevation angles-of-arrival (AoA), respectively. When $l=0$, the parameters correspond to those for the LoS channel. For simplicity, we set $\sigma_{(\tilde{u},u)} = \sigma_g, \forall (\tilde{u},u)$ in this paper. Omitting the index $(\tilde{u},u)$, the steering vector $\boldsymbol{a}(\theta_l, \phi_l)$ is defined as
\begin{equation}
\boldsymbol{a}(\theta_l, \phi_l) = {[1,e^{-j\frac{2\pi}{\lambda}d_l(2)},\dots,e^{-j\frac{2\pi}{\lambda}d_l(N)}]}^T,
\end{equation}
where $d_l(k)$ represents the propagation difference of the $l$-th path between the $(1,1)\to 1$-st port and the $(k_1,k_2)\to k$-th port. Accordingly, we have
\begin{equation}
d_l(k) \!=\! \frac{(k_1-1)W_1\lambda}{N_1-1}\sin \theta_l \cos \phi_l + \frac{(k_2-1)W_2\lambda}{N2-1} \cos \theta_l.
\end{equation}

A particular case pertains to the rich scattering scenario where $K=0$ and $N_p\to \infty$. In this context, the variables, $g^{(\tilde{u},u)}_k$, follow an identical complex Gaussian distribution, but the channels ${\{g^{(\tilde{u},u)}_k\}}_{\forall k}$ are correlated, characterized by a covariance matrix $ \mathbb{E}\left[ \boldsymbol{g}^{(\tilde{u},u)} (\boldsymbol{g}^{(\tilde{u},u)})^ \dag \right] = \sigma_g^2 \boldsymbol{\Sigma}$, defined as \cite{Khammassi2023}
\begin{equation}\label{Eq:corr}
\varrho_{k,\ell} \!\triangleq\! {\left[\boldsymbol{\Sigma}\right]}_{k,\ell}\!=\! j_0 \!\! \left(\! 2\pi \sqrt{{\left(\frac{k_1 \!-\! \ell_1}{N_1 \!-\! 1}\! W_1\right)}^2 \!\!+\! {\left(\frac{k_2 \!-\! \ell_2}{N_2 \!-\! 1}\! W_2\right)}^2}\right),
\end{equation}
where $j_0(\cdot)$ is the spherical Bessel function of the first kind.

In the conventional \emph{s}FAMA system \cite{wong2023sFAMA}, the port that maximizes the received signal-to-interference plus noise ratio (SINR) is designated as the optimal port for activation. Hence, for the $u$-th UE, we have 
\begin{equation}\label{Eq:sFAMA}
k^* = \arg \max_{k \in \Omega} \frac{\sigma_s^2 \lvert g^{(u,u)}_k\rvert ^2}{\sigma_s^2 \sum_{\substack{\tilde{u}=1\\\tilde{u}\neq u}}^{U} \lvert g^{(\tilde{u},u)}_k\rvert ^2 + \sigma_\eta^2},
\end{equation}
where $\Omega = \{1,\dots, N\}$ denotes the index set for the ports. 
More recently, \cite{coma2024slow,hong2025multiport,hong2025Downlink} have studied \emph{s}FAMA employing multi-port selection schemes within the digital domain. In these schemes, the number of activated ports is limited by the number of RF chains, such that $N^* = N_{\rm RF}$, with $N^*$ indicating the number of activated FAS ports.

From a different perspective, CUMA capitalizes on the combinatorial capability of analog modules. CUMA activates many ports and aggregates the received signals from these ports in the analog domain into a restricted number of RF chains for digital-domain processing~\cite{Wong2024cuma}. Consequently, we observe that $N^*\gg N_{\rm RF}$. Let us define the ports that are selected for activation and aggregation at the $i$-th RF chain having their indices stored in the set $\mathcal{K}_i$ ($i = 1,\dots,N_{\rm RF}$); hence $N^* = \lvert \underline{\mathcal{K}}\rvert$, where $\underline{\mathcal{K}} = \mathcal{K}_1 \cup \dots \cup \mathcal{K}_{N_{\rm RF}}$. CUMA employs certain strategies \cite[Section III-C]{Wong2024cuma} to determine the sets $\mathcal{K}_i$, and aggregates the signals directly in the analog domain as 
\begin{equation}
y_i = \bm{w}_{{\rm CUMA},i}^\dag \bm{r}^{(u)} = \sum_{k \in \mathcal{K}_i} r^{(u)}_k,
\end{equation}
where $\bm{w}_{{\rm CUMA},i} \triangleq \sum_{k \in \mathcal{K}_i} \bm{e}_k$ is the analog combining vector, and $\bm{e}_k$ is the standard basis vector. Then the estimated symbol with digital combining can be found as
\begin{equation} \label{eq:DigCom}
\tilde{s}_u = \bm{v}^\dag \bm{y}, 
\end{equation}
where $\bm{y} = [y_1,\dots,y_{N_{\rm RF}}]^T$ denotes the aggregated received signal vector, and $\bm{v} \in {\mathbb{C}^{N_{\rm RF}}}$ is the receiver digital beamforming vector. For benchmarking, the average received signal-to-noise ratio (SNR) for each UE is defined as $\Gamma \triangleq \sigma_s^2\sigma_g^2/\sigma_\eta^2$.

\begin{figure*}
\centering
\includegraphics[width = 0.8\linewidth]{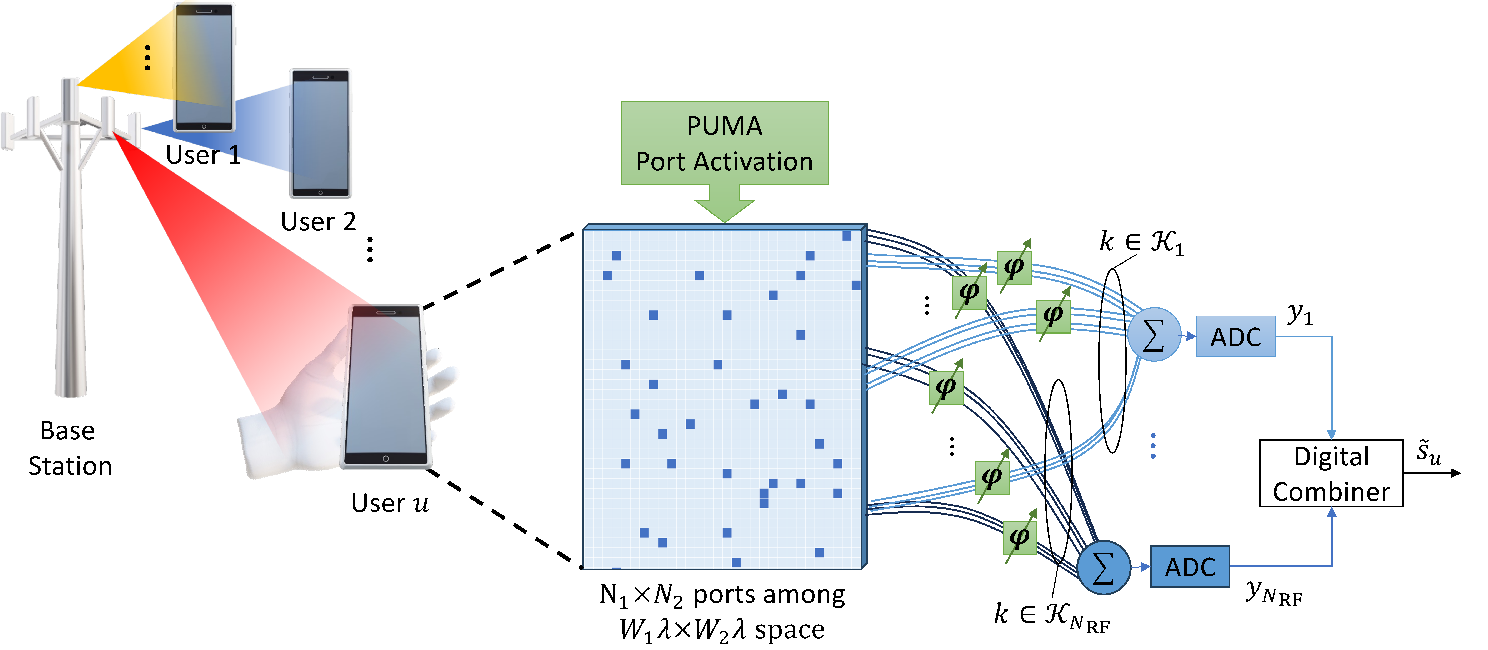}
\caption{Illustration of a downlink PUMA system.}\label{Fig:PUMA}
\end{figure*}

\section{PUMA}\label{sec:PUMA}
\subsection{Signal Model}
Here, we propose a new CUMA receiver architecture, referred to as PUMA. As illustrated in Fig.~\ref{Fig:PUMA}, a MIMO BS communicates with multiple PUMA-equipped UEs. The CSI, i.e., $\bm{g}_{(u,u)}$ encompassing the channel coefficients across all ports, is available at UE $u$. The principal distinction of PUMA from CUMA lies in its implementation of an analog domain phased array for signal aggregation.  Specifically, the ports selected for activation have their indices stored in $\mathcal{K}_i$ ($i=1, \dots, N_{\rm RF}$), and these activated ports undergo phase shifting during the aggregation process. Consequently, PUMA employs the analog combining weights constituted by constant-modulus complex elements for the ports within $\mathcal{K}_i$. This architectural approach mitigates significant power fluctuations. The methodology for determining the set $\mathcal{K}_i$ and configuring the phase shifts will be discussed comprehensively in Section \ref{subsec:WandK}.

With a large array of phase shifters, the signals from the activated ports given within the set $\mathcal{K}_i$ are phase shifted and subsequently aggregated to generate the received signal as
\begin{align}
y_i & = \bm{w}_i^\dag \bm{r}^{(u)}, \nonumber \\
& = \bm{w}_i^\dag \bm{g}^{(u,u)} s_{u} + \underbrace{\sum_{\substack{\tilde{u} = 1\\ \tilde{u}\neq u}}^U  \bm{w}_i^\dag  \bm{g}^{(\tilde{u},u)} s_{\tilde{u}} \!+\! \bm{w}_i^\dag \bm{\eta}^{(u)}}_{{z}_{i}},
\end{align}
where the interference-plus-noise component is denoted by ${z}_{i}$, the vector $\bm{w}_i = [w_{i,1},\dots,w_{i,N}]^T \in \mathbb{C}^{N\times 1}$, and $\lvert w_{i,k}\rvert = \mathbf{1}_{k\in \mathcal{K}_i}$ with $\mathbf{1}_{k\in \mathcal{K}}$ being an indicator function, taking the value $1$ when $k\in \mathcal{K}$, and $0$ otherwise.
For all $N_{\rm RF}$ RF chains, the received signals are represented by
\begin{align}
\bm{y} & = \mathbf{W}^\dag \bm{g}^{(u,u)} s_{u} + \bm{z},
\end{align}
where we have $\mathbf{W} = [\bm{w}_1,\dots,\bm{w}_{N_{\rm RF}}]\in \mathbb{C}^{N\times N_{\rm RF}}$, and $\bm{z} = [z_{1},\dots,z_{N_{\rm RF}}]^T$. With the knowledge of $\bm{g}^{(u,u)}$, the received signals $\bm{y}$ can be estimated and detected as in \eqref{eq:DigCom}.
Note that the phase shifters are employed within the analog domain during the aggregation process. The UE receiver only requires $N_{\rm RF}$ down-conversion RF chains, thereby significantly reducing {\color{black}the dominant RF-chain and baseband hardware complexity at the cost of additional low-power analog phase-control elements.}

\subsection{Phased Array and Port Activation Schemes}\label{subsec:WandK}
To design the port activation schemes jointly with the phase shifters in PUMA, our objective is to maximize the average received SINR. Without loss of generality, we assume a unit energy transmission, i.e., $\sigma_s^2 = 1$, which results in
\begin{equation}
\overline{\gamma} =  \frac{|\bm{v}^\dag \mathbf{W}^\dag \bm{g}^{(u,u)}|^2}{|\bm{v}^\dag \bm{z}|^2}.
\end{equation}
As a consequence, the optimal PUMA combining vectors that maximize the SINR are expressed as
\begin{equation}\label{eq:optW}
\mathbf{W}^* = \arg \max_{\mathbf{W}: |w_{i,k}| = \mathbf{1}_{k \in \mathcal{K}_i}} \overline{\gamma},
\end{equation}
which is non-convex and challenging to solve. This work prioritizes the identification of a feasible approximate solution over the pursuit of the absolute optimal solution. For this reason, {\color{black} we simplify the original design problem in two steps: First, rather than exploiting the instantaneous interference covariance, we approximate the interference-plus-noise covariance at the receiver by a spatially white covariance with a fixed variance, i.e., $\mathbb{E}[\bm{z}\bm{z}^\dag] = \sigma_I^2\mathbf{I}$. This approximation decouples the original design problem from the instantaneous co-user interference and allows the receiver design to primarily focus on enhancement of the desired signal. Secondly, we consider a relaxed problem without the constant-modulus constraint and obtain its optimal combining direction. The resulting unconstrained solution is then projected onto the constant-modulus feasible set to construct the PUMA design.}

{\color{black} The first approximation also follows the low-complexity receiver philosophy of the FAMA and CUMA systems. PUMA enables multiple access by leveraging the inherent advantages of the FAS at the receiver, rather than relying on complex multiuser detection algorithms. Therefore, we do not estimate the instantaneous interference covariance. Instead, the interference-plus-noise variance is approximated as $\mathbb{E}[|z_{i}|^2] = \sigma_I^2$. Utilizing the digital beamformer $\bm{v} = \mathbf{W}^\dag \bm{g}^{(u,u)} / \lVert \mathbf{W}^\dag \bm{g}^{(u,u)} \rVert$, the }average SINR is given by 
\begin{equation}
\overline{\gamma} \approx \frac{{\lVert \mathbf{W}^\dag \bm{g}^{(u,u)} \rVert}^2}{\sigma_I^2}.
\end{equation}
Thus, the problem in \eqref{eq:optW} can be simplified as
\begin{equation}\label{eq:optW2}
\mathbf{W}^* = \arg \max_{\mathbf{W}: |w_{i,k}| = \mathbf{1}_{k \in \mathcal{K}_i}} {\lVert \mathbf{W}^\dag \bm{g}^{(u,u)} \rVert}^2.
\end{equation}
{\color{black}The problem is thereby simplified to a desired signal power maximization problem. However,} the constraint $|w_{i,k}| = \mathbf{1}_{k \in \mathcal{K}_i}$ remains non-convex. 

{\color{black} In the second simplification, we consider a single-RF-chain scenario to obtain a tractable solution. In this scenario, all ports are activated, i.e., $N_{\rm RF} = 1$, $\mathcal{K}_1 = \Omega$, and the constant-modulus constraint on the coefficients $\bm{w}_i$ is temporarily relaxed. Under the spatially white interference-plus-noise approximation introduced above, the optimal unconstrained beamforming vector would be $\bm{u} = \bm{g}^{(u,u)}$. Since only the phases of the combining coefficients can be adjusted by the PUMA analog phase shifters, we construct a feasible approximate solution by projecting the unconstrained vector $\bm{u}$ onto the constant-modulus feasible set, i.e.,}
\begin{equation}\label{Eq:wijApp}
\tilde{w}_{i,k} = \begin{cases}
e^{j \angle g^{(u,u)}_k}, & k \in \mathcal{K}_i,\\
0, & \text{otherwise},
\end{cases}
\end{equation}
where $\tilde{w}_{i,k}$ denotes the $(i,k)$-th element of the approximated solution $\tilde{\mathbf{W}}$. With this, the average SINR can be estimated as
\begin{equation}
    \tilde{\gamma} = \frac{\sum_{i=1}^{N_{\rm RF}} \sum_{k\in \mathcal{K}_i} |g^{(u,u)}_k|}{\sigma_I^2}.
\end{equation}  
The phased array ensures that all ports can support the desired signal constructively. Ideally, the system can achieve the optimal performance by activating all ports. However, for practical reasons, it may be advantageous to reduce the cardinal number of $\mathcal{K}_i$, since the total number of ports, $N$, is typically very large. Similar to the CUMA scheme, we propose employing the scaling parameter $\rho$ ($0\leq \rho \leq 1$) and a maximum limit on the number of active ports $N_{\max}$ to facilitate this adjustment.

Specifically, we first employ $\rho$ to shortlist the ports as 
\begin{equation}\label{Eq:shortlist}
\tilde{\mathcal{K}} = \left\{k: \lvert g^{(u,u)}_k\rvert \geq  \rho \max_{\ell} \lvert g^{(u,u)}_\ell\rvert\right\}.
\end{equation}
Subsequently, a second criterion is implemented to restrict the cardinality of the set $\mathcal{K}_i$ to a maximum of $N_{\max}$. To achieve this, a random selection of up to $N_{\max}$ elements can be made to constitute $\mathcal{K}_i$, i.e., 
\begin{equation}\label{Eq:shortlistNm}
\left\{ \mathcal{K}_i \subseteq \tilde{\mathcal{K}} | \lvert \mathcal{K}_i\rvert \leq N_{\max}\right\},~ i= 1, \dots, N_{\rm RF}.
\end{equation}

\subsection{Mutual Coupling}
Closely spaced ports result in considerable mutual coupling. It is possible to incorporate mutual coupling effects so that the received signal vector is given by \cite{Clerckx2007MC}
\begin{equation}\label{Eq:yMC}
\tilde{\bm{y}} = \tilde{\mathbf{W}}^{\dag}\mathbf{\Gamma}_{\rm m} \mathbf{S}_{\underline{\mathcal{K}}} \bm{r}^{(u)},
\end{equation}
where $\tilde{\mathbf{W}}$ is the phase shifter matrix, $\mathbf{S}_{\underline{\mathcal{K}}}$ denotes a diagonal matrix with the diagonal element being ``1'' to indicate activated port and ``0'' otherwise according to the set $\underline{\mathcal{K}} = \mathcal{K}_1 \cup \dots \cup \mathcal{K}_{N_{\rm RF}}$, and $\mathbf{\Gamma}_{\rm m}$ represents the mutual coupling matrix as defined by
\begin{equation}
\mathbf{\Gamma}_{\rm m} = Z_T{(\mathbf{Z}+Z_T\mathbf{I})}^{-1},
\end{equation}
where $Z_T$ denotes the termination impedance with a typical value of $50~{\rm \Omega}$, and $\mathbf{Z}$ is the mutual impedance matrix. The exact expressions of the mutual impedance matrix $\mathbf{Z}$ relate to the antenna configuration and dipoles \cite{lee1988analysis,Singh2013MC}. Specifically, as the port density increases, the mutual impedance increases. Conversely, when the ports are not too close to each other, then we have $\mathbf{\Gamma}_{\rm m} \approx \mathbf{I}$, essentially without mutual coupling.

{\color{black} Under the considered coupling model,} PUMA operates similarly by utilizing the received vector in \eqref{Eq:yMC} and considering the effective channel ${\bm{h}} = \tilde{\mathbf{W}}^{\dag}\mathbf{\Gamma}_{\rm m} \mathbf{S}_{\underline{\mathcal{K}}} \bm{g}^{(u,u)}$ for detection in \eqref{eq:DigCom}, where the mutual coupling matrix, $\mathbf{\Gamma}_{\rm m}$, can be pre-computed offline.

\section{Performance Analysis}\label{sec:Perf}
\subsection{Assumptions and Notations}
This section aims to analyze the performance of PUMA under rich scattering while discarding mutual coupling effects. Since the performance depends on the interference experienced by each UE, our primary focus will be on characterizing the SIR by deriving its PDF prior to conducting the rate analysis. To facilitate a tractable analysis, we consider a specific case of PUMA where the receiver has a single RF chain with all ports activated, i.e., $N_{\rm RF} = 1$, $\rho = 0$, and $N_{\max} = N$. 

As the UEs are i.i.d., it suffices to focus on the performance of any representative UE within the analysis. As such, the UE index can be omitted, and the SINR for a UE is given by
\begin{align}
{\rm SINR} &= \frac{\sigma_s^2 \left(\sum_{k=1}^{N} \lvert g^{(u,u)}_k\rvert \right)^2}{\sigma_s^2 \sum_{\tilde{u} = 1 \atop \tilde{u} \neq u}^U \left| \sum_{k=1}^{N} g^{(\tilde{u},u)}_k e^{-j\angle g_k^{(u,u)}} \right|^2 + N\sigma_\eta^2}, \nonumber\\
& \approx \frac{\left(\sum_{k=1}^{N} \lvert g_k^{(u,u)}\rvert \right)^2}{\sum_{\tilde{u} = 1 \atop \tilde{u} \neq u}^U \left|\sum_{k=1}^{N} g_k^{(\tilde{u},u)} e^{-j\angle g_k^{(u,u)}} \right|^2} = {\rm SIR}, 
\end{align}
where the approximation is accurate in the interference-limited scenario, i.e., $\sigma_s^2 \gg \sigma_\eta^2$ or $U\gg 2$. To simplify the notations, the following definitions are introduced:
\begin{align}
X & \triangleq \left(\sum_{k=1}^{N} \lvert g_k^{(u,u)}\rvert \right)^2,\\
S_{\tilde{u}} & \triangleq \sum_{k=1}^{N} g_k^{(\tilde{u},u)} e^{-j\angle g_k^{(u,u)}}, \\ 
Y & \triangleq \sum_{\tilde{u} = 1 \atop \tilde{u} \neq u}^U \lvert S_{\tilde{u}}\rvert^2.
\end{align}
The SIR is then expressed as ${\rm SIR} = X/Y$.

\subsection{Main Results}
Here, we first derive the PDF of the SIR, and then analyze the average symbol error probability and the network rate.

\begin{theorem}\label{Theo:GaussianSqrtX}
{\color{black}For a FAS with a sufficiently large number of effective spatial degrees of freedom\footnote{\color{black}Here, the effective spatial degrees of freedom refer to the number of spatially distinct channel contributions whose mutual dependence becomes sufficiently weak with increasing physical separation.},} $\sqrt{X} = \sum_{k = 1}^N \lvert g_k^{(u,u)}\rvert $ is approximately Gaussian with mean
\begin{equation}\label{Eq:meanSqrtX}
\mu_1 \triangleq \mathbb{E}[\sqrt{X}] = \frac{N\sqrt{\pi}}{2}\sigma_g, 
\end{equation}
and variance
\begin{align}\label{Eq:varSqrtX}
\sigma_1^2 &\triangleq \var[\sqrt{X}] \nonumber \\
&= N\left(1-\frac{\pi}{4}\right)\sigma_g^2 \nonumber \\
&\quad+ \frac{\pi\sigma_g^2}{2}\sum_{\ell=2}^{N} \! \sum_{k=1}^{\ell-1} \left[{}_2F_1\left(-\frac{1}{2},-\frac{1}{2};1;\varrho_{k,\ell}^2\right) \!-\! 1\right], 
\end{align}
where ${}_2F_1(\cdot,\cdot;\cdot;\cdot)$ is the Gaussian hypergeometric function.
\end{theorem}

\begin{proof}
See Appendix \ref{app:proofTheoSqrtX}.
\end{proof}

Theorem \ref{Theo:GaussianSqrtX} indicates that the number of FAS ports at each UE enhances the performance by increasing the mean of desired channel gain in \eqref{Eq:meanSqrtX}. 
\begin{corollary}\label{Theo:pdfX}
The PDF of $X$ is given by
\begin{equation}\label{Eq:pdfX}
f_X(x) = \frac{1}{2\sigma_1^2} x^{-\frac{1}{4}} \mu_1^{\frac{1}{2}}e^{-\frac{x+\mu_1^2}{2\sigma_1^2}} I_{-\frac{1}{2}}\left(\frac{\mu_1}{\sigma_1^2}\sqrt{x}\right),~x\geq 0,
\end{equation}
where $I_r(u)$ denotes the modified Bessel function of the first kind and order $r$, $\mu_1 = \mathbb{E}[\sqrt{X}]$ is given by \eqref{Eq:meanSqrtX}, and $\sigma_1^2 = \var [\sqrt{X}]$ is given by \eqref{Eq:varSqrtX}.
\end{corollary}

\begin{proof}
According to \cite[Lemma 1]{Wong2024cuma}, given $\sqrt{X} \sim \mathcal{N}(\mu_1, \sigma_1^2)$ from Theorem \ref{Theo:GaussianSqrtX}, the PDF of $X$ is given by \eqref{Eq:pdfX}.
\end{proof}

\begin{theorem}\label{Theo:GaussianSu}
{\color{black}For a FAS with a sufficiently large number of effective spatial degree of freedom,} $S_{\tilde{u}}$ is approximately Gaussian with zero mean and 
\begin{equation}\label{Eq:varSu}
\sigma_2^2 \!\triangleq\! \var [S_{\tilde{u}}] \!=\! \sigma_g^2 N \!+\! \frac{\pi\sigma_g^2}{2} \sum_{\ell=2}^{N} \! \sum_{k=1}^{\ell-1} \varrho_{k,\ell}^2 {}_2F_1\left(\frac{1}{2},\frac{1}{2};2;\varrho_{k,\ell}^2\right).
\end{equation}
\end{theorem}

\begin{proof}
See Appendix \ref{app:proofTheoSu}.
\end{proof}

\begin{corollary}\label{Theo:pdfYY}
The variable $\tilde{Y} = Y/\sigma_2^2$ follows the Gamma distribution: $\tilde{Y}\sim \Gamma(U-1,1)$, and its PDF is given by
\begin{equation}\label{Eq:pdfYY}
f_{\tilde{Y}}(\tilde{y}) = \frac{\tilde{y}^{U-2}e^{-\tilde{y}}}{\Gamma(U-1)},
\end{equation}
where $\Gamma(\cdot)$ is the Gamma function.
\end{corollary}

\begin{proof}
According to Theorem \ref{Theo:GaussianSu}, $\tilde{S}_{\tilde{u}} = S_{\tilde{u}}/\sigma_2$ is a standard complex normal variable. Thus, ${\rvert\tilde{S}_{\tilde{u}}\lvert}^2$ follows an exponential distribution with mean $1$. Also, $\tilde{Y} = \sum_{\tilde{u}=1\atop \tilde{u}\neq u}^{U} {\rvert\tilde{S}_{\tilde{u}}\lvert}^2$ is the sum of $(U-1)$ i.i.d.~exponential random variables, which follows a Gamma distribution with a shape parameter of $(U-1)$ and a scale parameter of $1$, as described by the PDF in \eqref{Eq:pdfYY}.
\end{proof}

Corollary~\ref{Theo:pdfYY} indicates that the interference escalates with the number of UEs, as the expectation $\mathbb{E}[\tilde{Y}] = U-1$ increases concomitantly with $U$. An increase in $U$ will consequently lead to a decline in performance for each UE. 

\begin{theorem}\label{Theo:pdfZ}
The PDF of $Z = X/\tilde{Y}$ is given by \eqref{Eq:pdfZ}, shown on the top of next page, where $\mu_1$ is given by \eqref{Eq:meanSqrtX}, $\sigma_1^2$ is given by \eqref{Eq:varSqrtX}, and $\mathcal{M}_{a,b}(t)$ represents the Whittaker $M$ function. 
\end{theorem}

\begin{figure*}
\begin{equation}\label{Eq:pdfZ}
f_{Z}(z) = \frac{\Gamma(U-\frac{1}{2})}{\Gamma(U-1)\Gamma(\frac{1}{2})} z^{-\frac{3}{4}}\mu_1^{-\frac{1}{2}} e^{-\frac{\mu_1^2(4\sigma_1^2+4)}{4\sigma_1^2(2\sigma_1^2+z)}} {\left(\frac{z}{2\sigma_1^2}+1\right)}^{-U+\frac{3}{4}} \mathcal{M}_{-U+\frac{3}{4},-\frac{1}{4}}\left(\frac{\mu_1^2 z}{2\sigma_1^2 (2\sigma_1^2+z)}\right),~z \geq 0,
\end{equation}
\hrulefill
\end{figure*}

\begin{proof}
See Appendix \ref{app:proofTheopdfZ}.
\end{proof}

\begin{lemma}
The average bit error rate (BER) of a typical PUMA UE using binary phase shift keying (BPSK) is given by
\begin{equation}\label{Eq:ABERBPSK}
{p}_{e,{\rm BPSK}} = \int_0^\infty Q\left(\sqrt{\frac{2z}{\sigma_2^2}}\right)f_Z(z)dz,
\end{equation}
whereas the BER for $M$-ary quadrature amplitude modulation (QAM) is given by
\begin{equation}\label{Eq:ABERQAM}
{p}_{e,{\rm QAM}} \!=\! \frac{4}{m} \! \left(1\!-\!\frac{1}{\sqrt{M}}\right) \!\int_0^\infty\!\! Q\left(\sqrt{\frac{3z}{(M-1)\sigma_2^2}}\right)f_Z(z)dz,
\end{equation}
where $m = \log_2 M$, $Q(\cdot)$ represents the Gaussian Q-function, and $\sigma_2^2$ is given by \eqref{Eq:varSu}.
\end{lemma}

\begin{proof}
With PUMA, the received signal at a typical UE can be modeled as
\begin{equation}
y = \sqrt{X}s+\tilde{\eta},
\end{equation}
where $\tilde{\eta}$ is the resulting noise with zero mean and variance of $Y\sigma_s^2$. Consequently, the BER for this specific channel employing BPSK is determined as
\begin{equation}\label{Eq:BERBPSK}
p_e(X,Y) \overset{(a)}{\approx} Q\left(\sqrt{\frac{2X\sigma_s^2}{Y\sigma_s^2}}\right) \overset{(b)}{=} Q\left(\sqrt{\frac{2Z}{\sigma_2^2}}\right),
\end{equation}
and that using $M$-QAM is given by
\begin{align}\label{Eq:BERQAM}
p_e(X,Y) & \overset{(a)}{\approx} \frac{4}{m} \left(1-\frac{1}{\sqrt{M}}\right) Q\left(\sqrt{\frac{3X\sigma_s^2}{(M-1)Y\sigma_s^2}}\right)\nonumber\\
& \overset{(b)}{=}\frac{4}{m} \left(1-\frac{1}{\sqrt{M}}\right) Q\left(\sqrt{\frac{3Z}{(M-1)\sigma_2^2}}\right),
\end{align}
where $(a)$ employs the approximated BER expression for BPSK and $M$-QAM \cite{khan2014anew}, and $(b)$ adopts the definition presented in Corollary \ref{Theo:pdfYY}. Finally, the average BER in \eqref{Eq:ABERBPSK} or \eqref{Eq:ABERQAM} is derived by averaging \eqref{Eq:BERBPSK} or \eqref{Eq:BERQAM} over the PDF specified in \eqref{Eq:pdfZ}, which completes the proof.
\end{proof}

Based on the derived results, the PUMA network can be assessed using several key performance indicators comparable to those outlined in \cite[Corollaries 2--4]{Wong2024cuma}. First, the data rate for binary symmetric channels is given by 
\begin{equation}\label{Eq:Rate}
R = U\left[1+p_e\log_2 p_e +(1-p_e)\log_2 (1-p_e)\right],
\end{equation} 
where $p_e$ is given by \eqref{Eq:ABERBPSK} or \eqref{Eq:ABERQAM}, depending on the modulation constellation employed in the system. Secondly, the ergodic rate can be determined as 
\begin{equation}
C_e = U\int_0^\infty \log_2 \left(1+\frac{z}{\sigma_2^2}\right)f_Z{(z)}dz,
\end{equation}
where $\sigma_2^2$ is specified in \eqref{Eq:varSu}. Finally, the $\varepsilon$-outage rate can be obtained by
\begin{equation}
C_{\rm out} = U \log_2(1+\gamma_{\rm th}),
\end{equation}
where $\gamma_{\rm th}$ satisfies
\begin{equation}
\int_0^{\sigma_2^2\gamma_{\rm th}} f_Z(z)dz = \varepsilon.
\end{equation}

\section{Simulation Results}\label{Sec:Sim}
\begin{table}[t]
\begin{center}
\vspace{-2mm}
\caption{Simulation Parameters}\label{Tab:SimPara}
\resizebox{.9\columnwidth}{!}{
\begin{tabular}{l|l|l|l}
        \hline
        \multicolumn{2}{l|}{\textbf{Parameter}}     & \multicolumn{2}{l}{\textbf{Value}} \\ \hline\hline
        \multicolumn{2}{l|}{Carrier frequency, $f_c$}& {$6 ~{\rm GHz}$}& {$26 ~{\rm GHz}$} \\ \hline
        \multicolumn{2}{l|}{Wavelength, $\lambda$}   & {$5 ~{\rm cm}$}  & {$1.15 ~{\rm cm}$} \\ \hline
        \multicolumn{2}{l|}{FAS size, $W$}          & {$3\lambda \times 1.6 \lambda$}&  {$13\lambda \times 7\lambda$} \\ \hline
        \multicolumn{2}{l|}{Average SNR, $\Gamma = \sigma_g^2\sigma_s^2/\sigma_\eta^2$}& \multicolumn{2}{l}{$50~{\rm dB}$}\\ \hline
        \multicolumn{2}{l|}{Demapping scheme}       & \multicolumn{2}{l}{Max-Log-MAP} \\ \hline
        \multicolumn{2}{l|}{Decoding scheme}        & \multicolumn{2}{l}{Min-Sum}      \\ \hline
        \multirow{2}{*}{\makecell{Finite-scattering \\ Channels}} & Rice factor, $K$    & {$0$}   & {$7$} \\ \cline{2-4}
        & No. of NLOS paths, $N_p$& {$50$}        & {$2$} \\ \hline
        \multirow{2}{*}{Dipole} & Length & \multicolumn{2}{l}{$0.5\lambda$}\\ \cline{2-4}
        & Width &   \multicolumn{2}{l}{$0.005\lambda$} \\ \hline\hline
    \end{tabular}}
\end{center}
\end{table}

\begin{table}
\begin{center}
\begin{threeparttable}
\caption{Compactness or Port Density of FAS ($N_1\times N_2$)}\label{Tab:FASComp}
\begin{tabular}{c|c|c|c}
    \hline
    \multicolumn{2}{c|}{ } & $6~{\rm GHz}$ & $26~{\rm GHz}$ \\ \hline
    \multicolumn{2}{r|}{FAS Size, $W$}& $3\lambda \times 1.6\lambda$    & $13\lambda \times 7\lambda$\\ \hline\hline
    \multirow{4}{*}{Compactness}& 
      Sparse\tnote{\ddag}       & $4\times 4$   & $14\times 15$ \\ \cline{2-4}
    & Not compact\tnote{\ddag}  & $7\times 4$   & $27\times 15$ \\ \cline{2-4}
    & Compact\tnote{\ddag}      & $16\times 4$  & $66\times 15$\\ \cline{2-4}
    & Very compact\tnote{\ddag} & $61\times 4$  & $261\times 15$\\ \hline
\end{tabular}
\begin{tablenotes}
\item[\ddag] The sparse case has a minimum spacing of $\lambda$. The not-compact case has a minimum spacing of $0.5\lambda$. The compact case has a minimum spacing of $0.2\lambda$. The very compact case has a minimum spacing of $0.05\lambda$. 
\end{tablenotes}
\end{threeparttable}
\end{center}
\end{table}

This section presents simulation results aimed at evaluating the performance of PUMA. The parameters used in the simulations are detailed in Table~\ref{Tab:SimPara}. The actual dimensions of the 2D-FAS at each UT are set to $15~{\rm cm}\times 8~{\rm cm}$, approximating the size of a typical mobile handset. In Table \ref{Tab:FASComp}, we detail the FAS configurations considered. The case at $6~{\rm GHz}$ frequency corresponds to the 5G Mid-band with a bandwidth of $10~{\rm MHz}$, whereas the $26~{\rm GHz}$ case pertains to the millimeter-wave band with a bandwidth of $50~{\rm MHz}$. We consider interference-limited scenarios where the average SNR $\Gamma = \sigma_g^2\sigma_s^2/\sigma_\eta^2$ is set to be $50~{\rm dB}$. We focus on the case with a known channel vector $\bm{g}^{(u,u)}$ for the typical UE $u$ and unknown channel vectors for interferers, $\boldsymbol{g}^{(\tilde{u},u)}$ ($\forall \tilde{u} \neq u$), when performing PUMA.

\subsection{Theoretical Performance under Rich Scattering}
Here, we consider rich scattering scenarios without mutual coupling effects and utilize the analytical results to evaluate the performance of PUMA. Fig.~\ref{Fig:PDF} shows the PDFs from Monte Carlo simulations and the analytical expression in \eqref{Eq:pdfZ} for various combinations of the system parameters. The results confirm the correctness and accuracy of the analytical results.

\begin{figure}[]
\centering
\includegraphics[width = \linewidth]{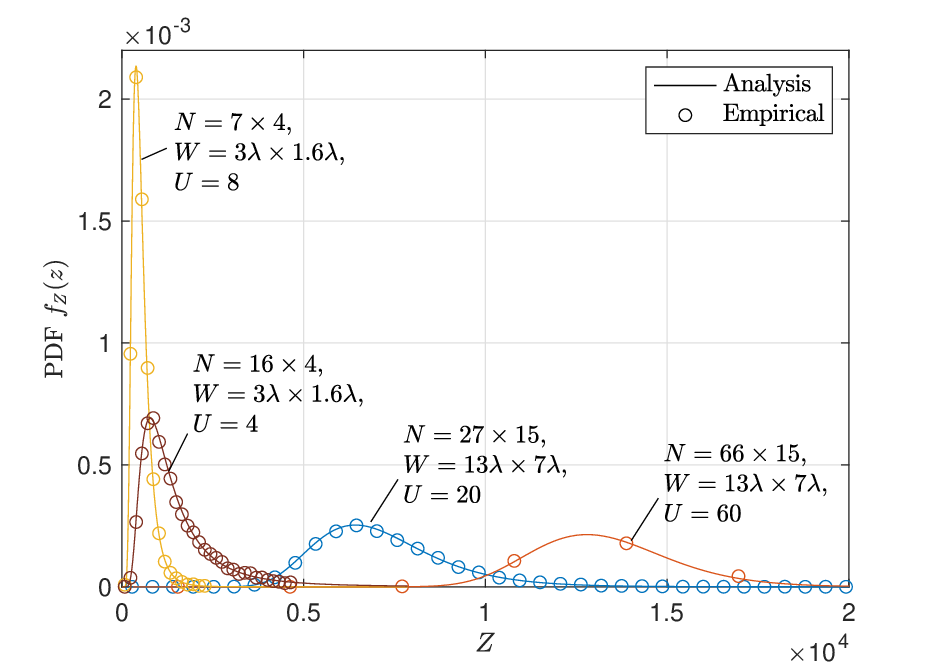}
\caption{Empirical and analytical results of PDFs for $Z$ in PUMA.}\label{Fig:PDF}
\end{figure}

Fig.~\ref{Fig:TheoRvsU} illustrates the average data rate of PUMA using QPSK. Specifically, Fig.~\ref{subfig:TheoRvsU_6GHz} presents the results at $6~{\rm GHz}$, whereas Fig. \ref{subfig:TheoRvsU_26GHz} pertains to $26~{\rm GHz}$. The average data rate is calculated as in \eqref{Eq:Rate} with an analytical expression in \eqref{Eq:ABERQAM} or through Monte Carlo simulations. These empirical simulation results confirm the precision of the data rate analysis.

\begin{figure}
\centering
\subfigure[$f_c = 6~{\rm GHz}$]{\includegraphics[width = 0.98\linewidth]{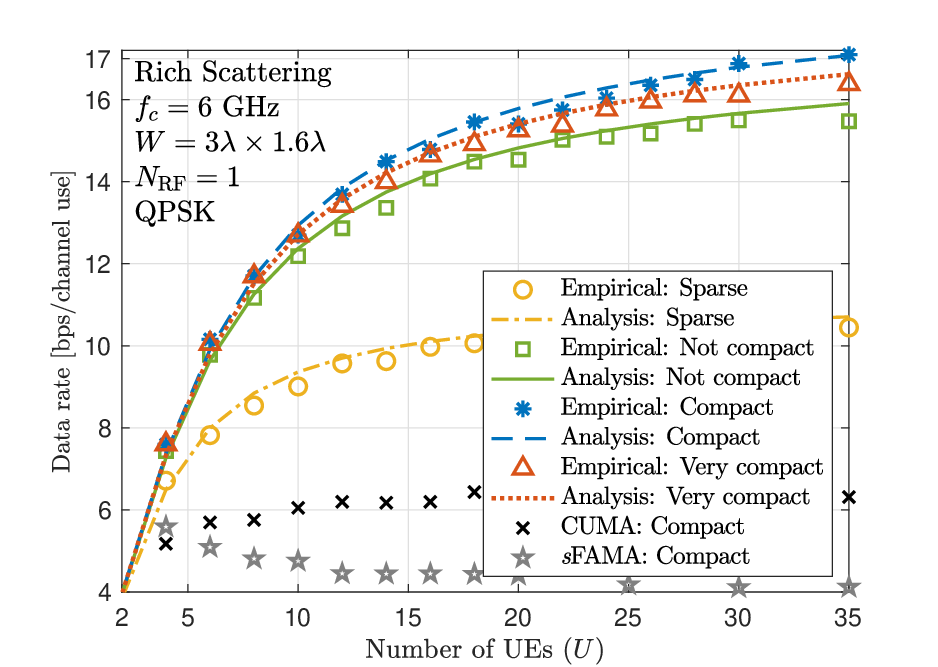}\label{subfig:TheoRvsU_6GHz}}\\
\subfigure[$f_c = 26~{\rm GHz}$]{\includegraphics[width = 0.98\linewidth]{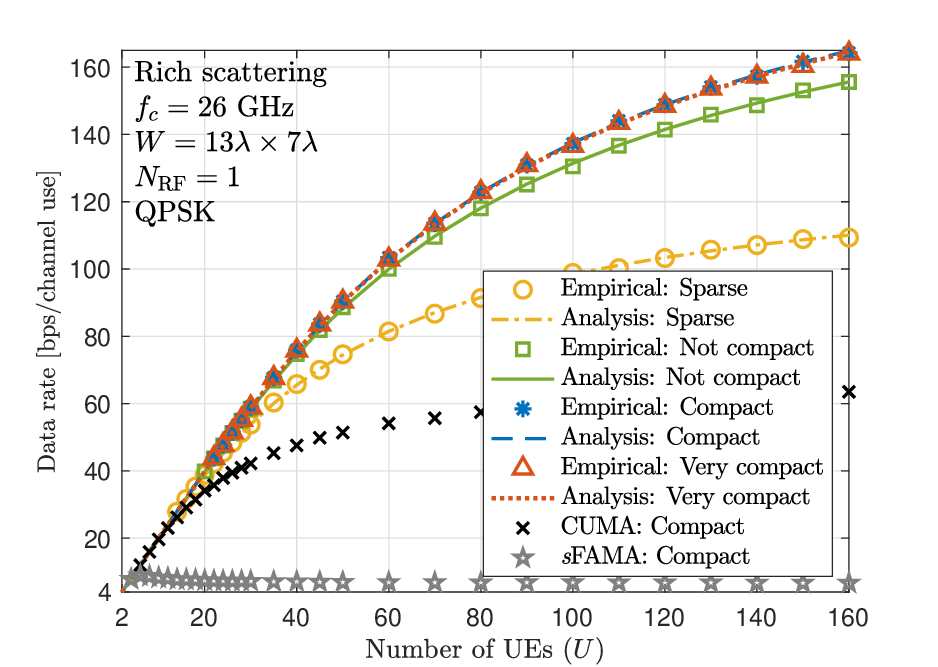}\label{subfig:TheoRvsU_26GHz}}
\caption{Average data rate performance of PUMA with QPSK modulation in the absense of mutual coupling and under rich scattering.}\label{Fig:TheoRvsU}
\end{figure}

At $6~{\rm GHz}$, as illustrated in Fig.~\ref{subfig:TheoRvsU_6GHz}, PUMA is capable of achieving only a relatively low data rate. This limitation is attributed to the small normalized FAS size at this frequency, which leads to highly correlated received signals at the FAS ports. Conversely, when the normalized size of FAS increases in Fig.~\ref{subfig:TheoRvsU_26GHz} (i.e., at a higher frequency), the data rate improves significantly. The system can achieve a remarkable data rate, and the rate increases monotonically with the number of UEs.

With a fixed normalized FAS size, increasing the number of antenna ports results in a noteworthy performance enhancement, particularly when the antenna ports are not densely packed (i.e., when the minimum spacing $\geq 0.5 \lambda$). However, more dense port packing yields marginal gains for PUMA under rich sacttering conditions. Specifically, when compared to the case with a minimum spacing of $0.5\lambda$, the rate gain from a more compact arrangement with a minimum spacing of $0.2\lambda$ is marginal. Moreover, further increasing the number of ports by reducing the minimum spacing to $0.05\lambda$ results in negligible gain, or even performance loss at $6~{\rm GHz}$.

\begin{figure*}
\centering
\subfigure[$f_c = 6~{\rm GHz}$ (Not compact)]{\includegraphics[width = 0.49\linewidth]{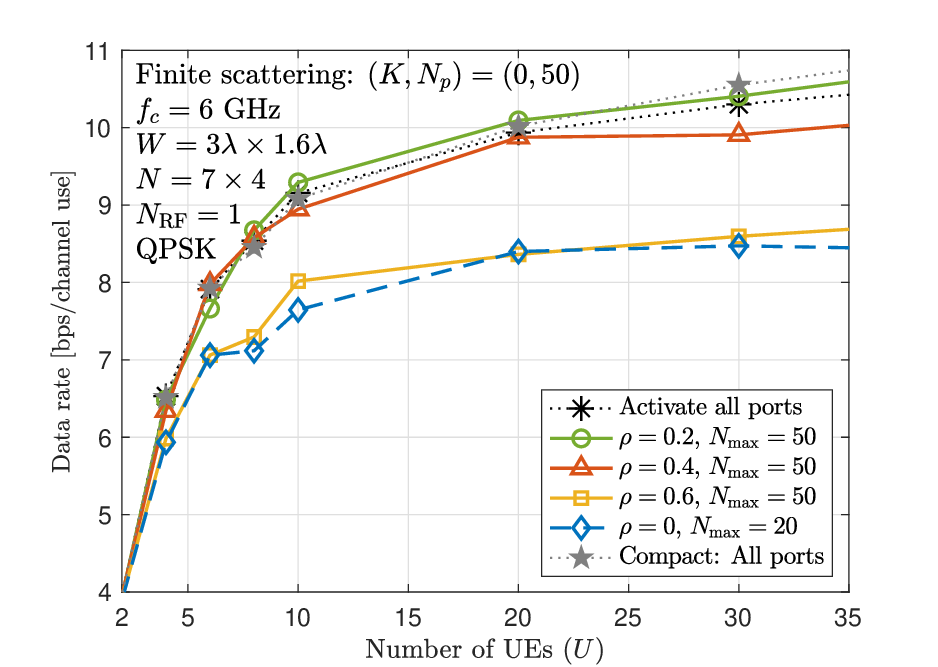}\label{subfig:RvsUFS_6GHzN}}
\subfigure[$f_c = 6~{\rm GHz}$ (Compact)]{\includegraphics[width = 0.49\linewidth]{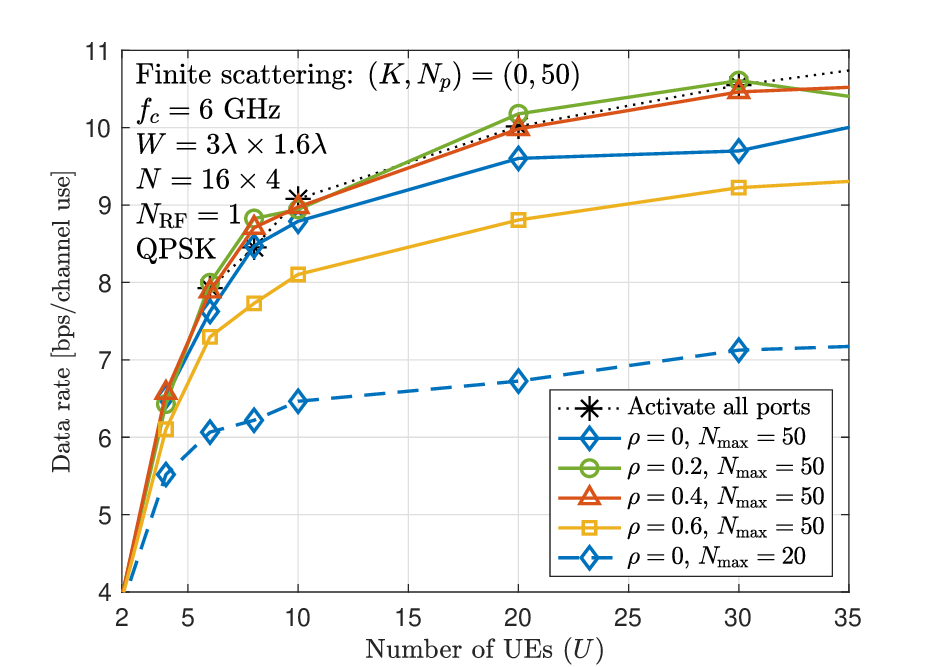}\label{subfig:RvsUFS_6GHzC}}\\
\subfigure[$f_c = 26~{\rm GHz}$ (Not compact)]{\includegraphics[width = 0.49\linewidth]{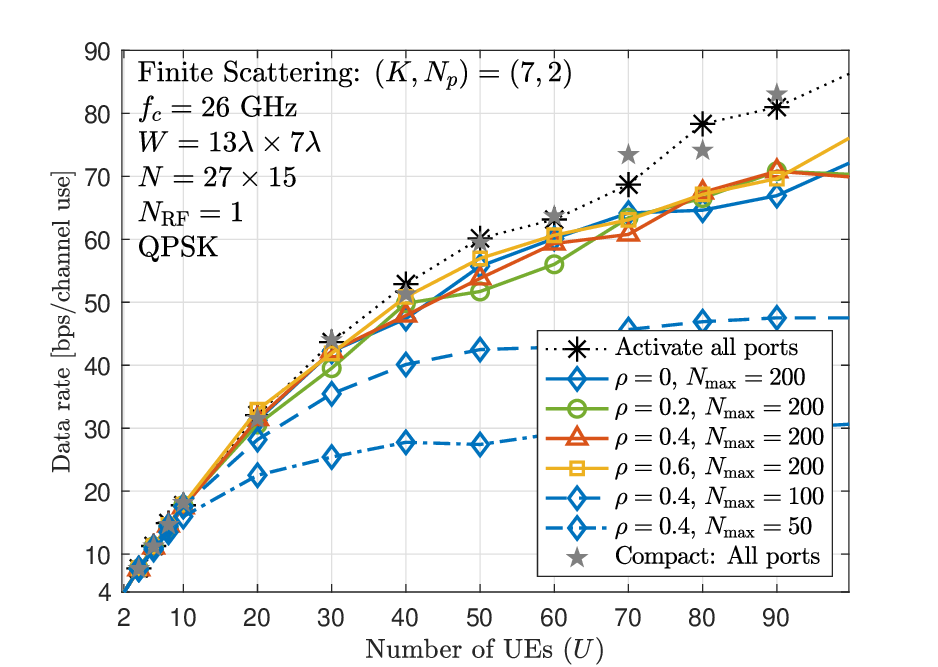}\label{subfig:RvsUFS_26GHzN}}
\subfigure[$f_c = 26~{\rm GHz}$ (Compact)]{\includegraphics[width = 0.49\linewidth]{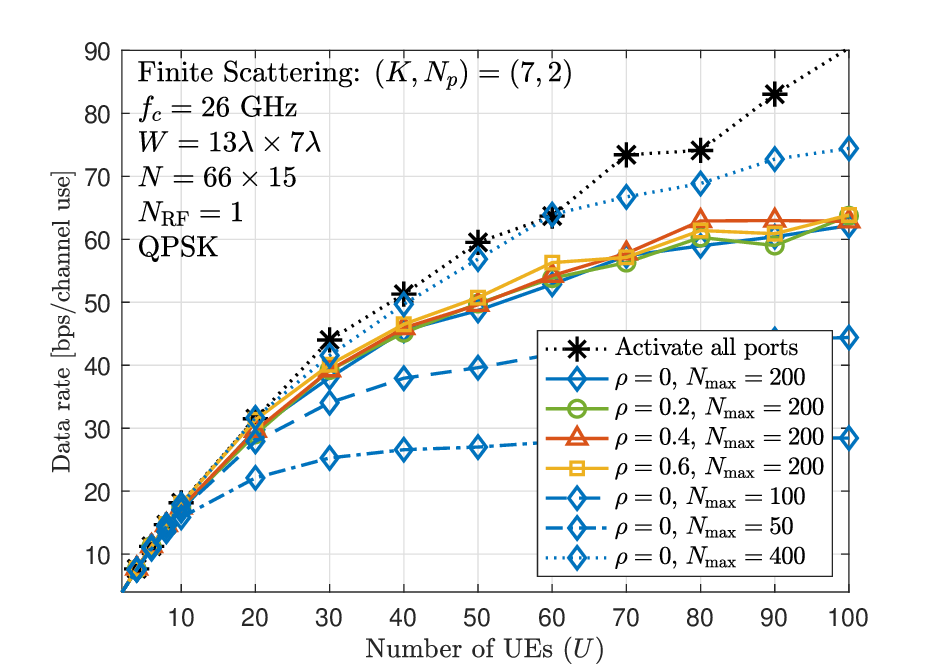}\label{subfig:RvsUFS_26GHzC}}
\caption{The rate performance of PUMA with different shortlist thresholds, $(\rho, N_{\max})$, under finite scattering with mutual coupling effects.}\label{Fig:RvsUFS}
\end{figure*}

\begin{figure*}
\centering
\subfigure[$f_c = 6~{\rm GHz}$ (Compact)]{\includegraphics[width = 0.49\linewidth]{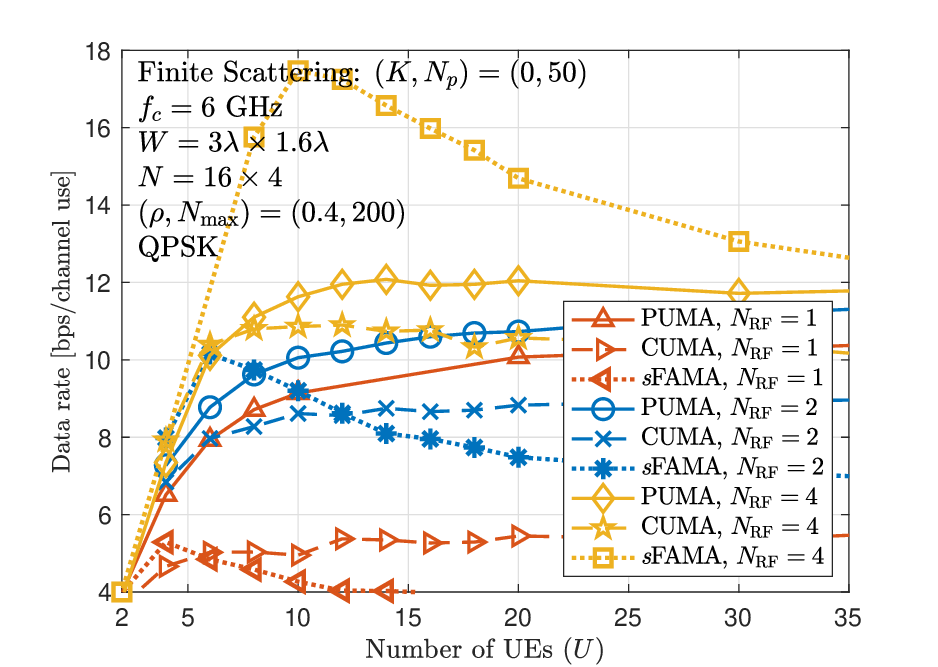}\label{subfig:DiffNRF_6GHzC}}
\subfigure[$f_c = 26~{\rm GHz}$ (Compact)]{\includegraphics[width = 0.49\linewidth]{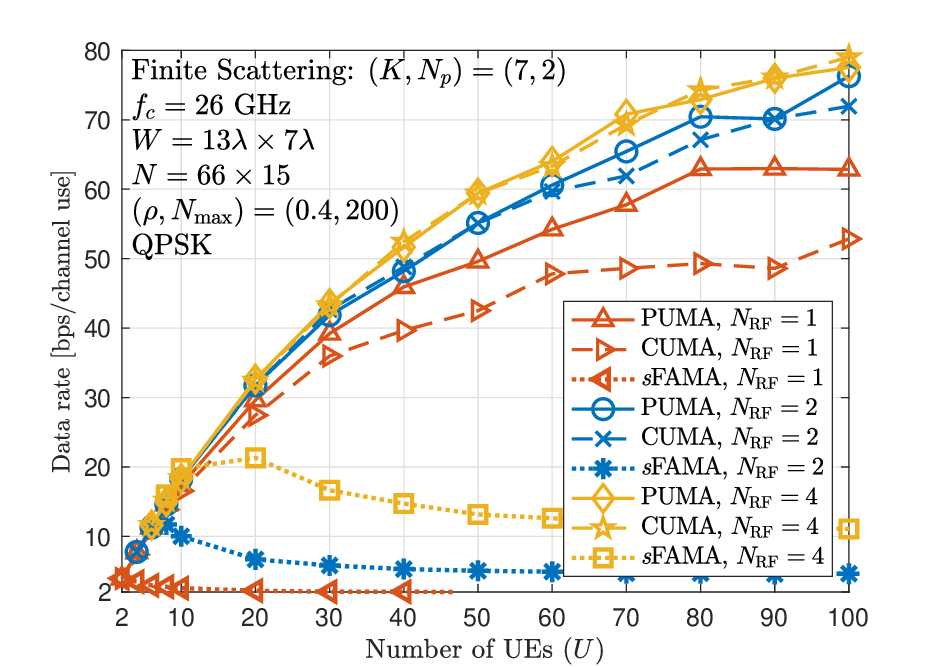}\label{subfig:DiffNRF_26GHzC}}
\caption{The rate performance of PUMA with different $N_{\rm RF}$, under finite scattering with mutual coupling effects, compared with CUMA and \emph{s}FAMA.}\label{Fig:RvsUFS_DiffNRF}
\end{figure*}

The empirical results of CUMA \cite{Wong2024cuma} and \emph{s}FAMA \cite{wong2023sFAMA} are depicted in black and gray in Fig. \ref{Fig:TheoRvsU}, respectively, for comparison. Note that each UE utilizes only a single RF chain, so the CUMA receiver employs the port activation metric of the in-phase channels exclusively with $\rho = 0$ and $N_{\max} = N$. As demonstrated, the proposed PUMA significantly outperforms both CUMA and \emph{s}FAMA under rich scattering. This gain is attributable to the channel-aware phase shifters during signal aggregation in PUMA. The superior performance of CUMA relative to \emph{s}FAMA, as demonstrated in \cite{Wong2024cuma}, arises from the fact that leveraging a large number of ports increases the efficacy of interference averaging. Moreover, the improved performance of PUMA in comparison to CUMA can be ascribed to the incorporation of phase shifters, which proactively mitigates interference during the aggregation process.

\subsection{Rate Performance under Finite Scattering}
This subsection presents the simulation results under finite scattering with mutual coupling. The channel parameters $(K,N_p)$ shown in Table \ref{Tab:SimPara} are selected carefully to reflect the characteristics of the respective frequency. A rectangular array of dipoles is assumed to form the 2D-FAS at each UE with row and column spacing depending on the compactness of FAS. The termination impedance is set to be $Z_T = 50~{\rm \Omega}$.

The results in Fig.~\ref{Fig:RvsUFS} are provided to investigate the average data rate performance of PUMA at various frequencies with different value of $(\rho,N_{\max})$. Figs.~\ref{subfig:RvsUFS_6GHzN} and \ref{subfig:RvsUFS_6GHzC} illustrate the rate results at $6~{\rm GHz}$. In Fig.~\ref{subfig:RvsUFS_6GHzN}, the results assume that the receiver has antenna ports with a minimum spacing of $0.5\lambda$, whereas Fig.~\ref{subfig:RvsUFS_6GHzC} considers the compact case with a minimum spacing of $0.2\lambda$. Similar to the observations in the rich scattering environment, PUMA exhibits suboptimal performance at $6~{\rm GHz}$ due to the limited effective FAS size. In the not-compact case depicted in Fig.~\ref{subfig:RvsUFS_6GHzN}, the total number of ports, $N = 28$, is relatively small, so the threshold $N_{\max} = 50 > N$ does not influence the port activation. The performance is highly sensitive to variations in $N_{\max}$ when $N$ is small. As can be observed, the performance deteriorates when $N_{\max}$ is reduced to $20$. Regarding the other shortlist parameter, $\rho$, varying from $0$ to $0.4$, is neither particularly advantageous nor harmful to the rate performance. However, the performance declines when $\rho$ becomes excessively large (e.g., $\rho = 0.6$), as this can reduce the number of ports available for constructive reception. In the compact case illustrated in Fig.~\ref{subfig:RvsUFS_6GHzC}, the results demonstrate that the rate increases with $N_{\max}$. Small values of $\rho$, such as $\rho = 0.2$ or $0.4$, provides marginal rate gain over $\rho = 0$, but further increasing $\rho$ to $0.6$ is detrimental to overall performance.

Figs.~\ref{subfig:RvsUFS_26GHzN} and \ref{subfig:RvsUFS_26GHzC} show the results at $26~{\rm GHz}$ for the not compact and compact scenarios, respectively. The results further reveal that increasing $N_{\max}$ enhances rate performance, but variations in $\rho$ at $26~{\rm GHz}$ have a minimal impact. Note that the benefits of increasing $N_{\max}$ are limited by a maximum threshold, as the rate at $N_{\max}=200$ or $400$ closely approximates the ideal rate achieved by activating all ports in the non-compact or compact cases, respectively, despite fewer than half of the ports being activated. Additionally, it is evident that relative to the lower frequencies in Figs.~\ref{subfig:RvsUFS_6GHzN} and \ref{subfig:RvsUFS_6GHzC}, the electrical size of the FAS at the receiver substantially affects performance. Maintaining identical physical dimensions, a larger normalized size at higher frequencies facilitates the accommodation of more UEs within the same channel. Also, the results of the compact cases with all ports activated are depicted in Figs.~\ref{subfig:RvsUFS_6GHzN} and \ref{subfig:RvsUFS_26GHzN} in gray. The results indicate that packing denser FAS ports at the UE is neither advantageous nor detrimental, as the performance in both scenarios remains remarkably consistent. This suggests that the effects of mutual coupling do not present any issues in PUMA. 

Fig.~\ref{Fig:RvsUFS_DiffNRF} evaluates the variations in performance with differing numbers of RF chains, $N_{\rm RF}$. Benchmark comparisons are provided with results for CUMA \cite{Wong2024cuma} and \emph{s}FAMA \cite{wong2023sFAMA}. For \emph{s}FAMA with multiple RF chains, the incremental port selection (IPS) method in \cite{hong2025multiport} are adopted. We first analyze the results at $6~{\rm GHz}$ in Fig. \ref{subfig:DiffNRF_6GHzC}. Although the data rate is comparatively lower at this frequency, it is evident that PUMA outperforms CUMA. Notably, when $N_{\rm RF} = 1$, PUMA delivers nearly twice the data rate relative to CUMA. The performance gains diminish as $N_{\rm RF}$ increases. Note that increasing $N_{\rm RF}$ yields only marginal performance enhancements for PUMA, attributable to the large values of $N_{\max}$. As such, the random activation method in \eqref{Eq:shortlistNm} effectively explores all constructive ports. Further increases in $N_{\rm RF}$ offer limited additional spatial diversity, resulting in only marginal improvements. Conversely, the IPS-based \emph{s}FAMA achieves its maximum data rate when $N_{\rm RF}=4$ at $6~{\rm GHz}$, consistent with the findings reported in \cite{hong2025multiport}. Fig.~\ref{subfig:DiffNRF_26GHzC} illustrates the data rate at $26~{\rm GHz}$. At this high frequency, the performance of \emph{s}FAMA declines. The proposed PUMA outperforms \emph{s}FAMA within the range of $N_{\rm RF}$ from $1$ to $4$. Additionally, PUMA achieves a higher data rate than CUMA with a single RF chain. Nevertheless, the performance advantage of PUMA over CUMA becomes marginal as $N_{\rm RF}$ increases. Overall, the proposed PUMA shows its performance gains most notably when the receiver has only a single RF chain. This gain is achieved through a group of channel-aware phase shifters engaged in signal aggregation. With an adequately large size of FAS, PUMA can support dozens of UEs on the same channel. {\color{black}Consistent with CUMA and \emph{s}FAMA, PUMA does not require instantaneous CSI at the BS for precoding nor complex multiuser detection at the UE.}

\subsection{Coded Performance}
\begin{figure*}
\centering
\subfigure[Rich Scattering]{\includegraphics[width = 0.495\linewidth]{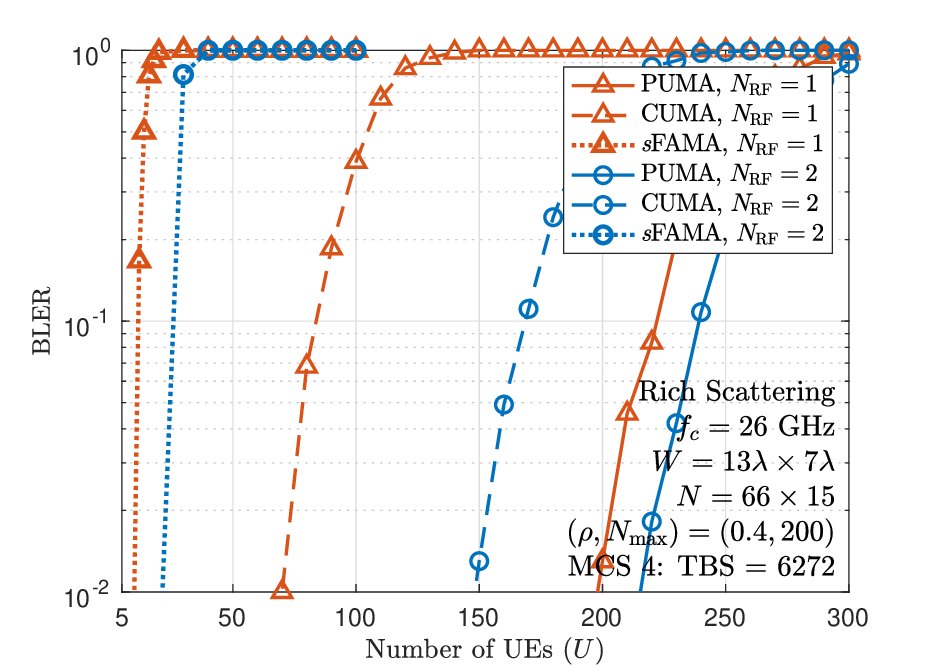}\label{subfig:BLERvsU_RS}}
\subfigure[Finite Scattering]{\includegraphics[width = 0.495\linewidth]{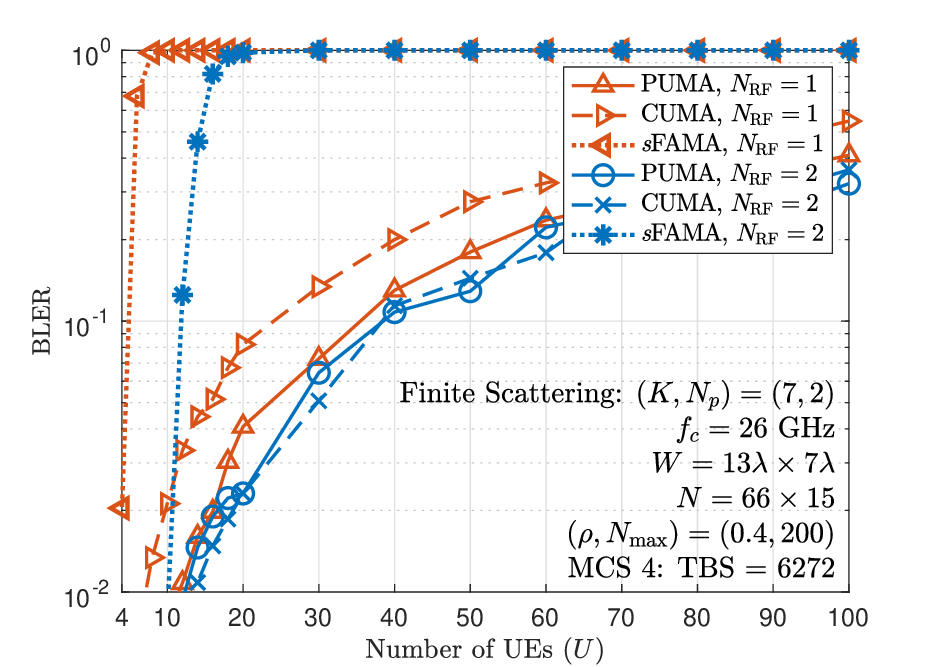}\label{subfig:BLERvsU_FS}}
\caption{The BLER performance of PUMA against $U$, under (a) rich scattering, and (b) finite scattering with mutual coupling.}\label{Fig:BLERvsU}
\end{figure*}

The proposed PUMA is subsequently evaluated within the coded system using 5G NR modulation and coding schemes (MCS) \cite{38214}. In this subsection, we focus on the $26~{\rm GHz}$ frequency with a compact configuration FAS at each UE. 

\begin{figure*}
\centering
\subfigure[Rich Scattering]{\includegraphics[width = 0.495\linewidth]{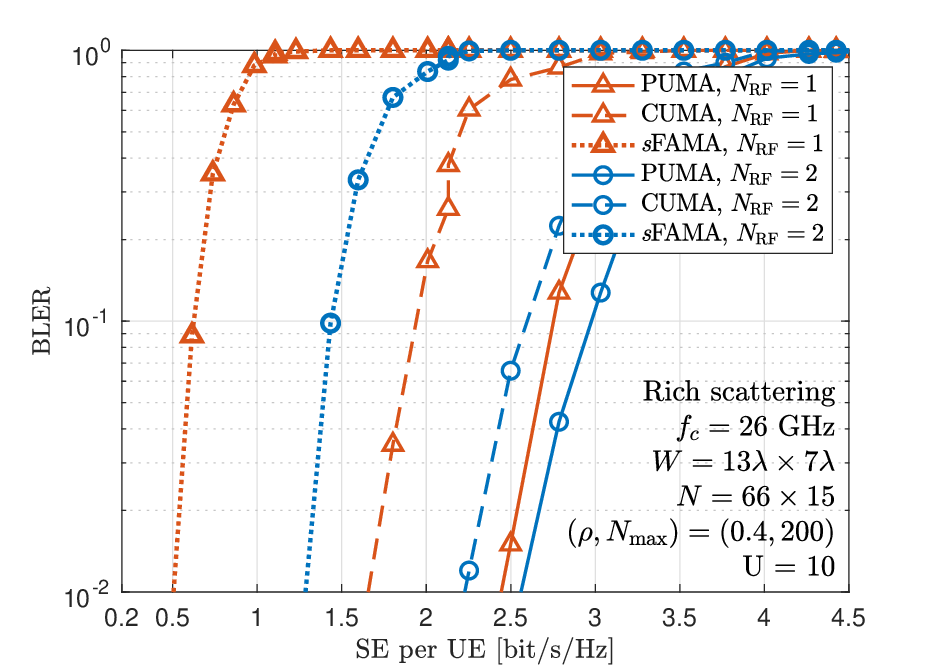}\label{subfig:BLERvsSE_RS}}
\subfigure[Finite Scattering]{\includegraphics[width = 0.495\linewidth]{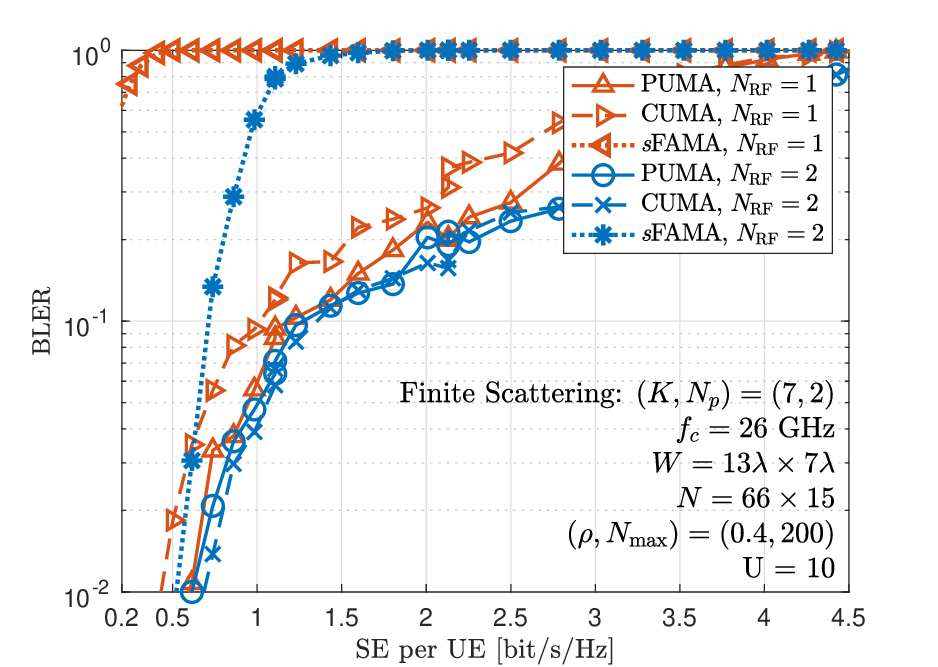}\label{subfig:BLERvsSE_FS}}
\caption{The BLER performance of PUMA agains spectral efficiency per UE, under (a) rich scattering, and (b) finite scattering with mutual coupling.}\label{Fig:BLERvsSE}
\end{figure*}

\begin{figure*}
\centering
\subfigure[Rich Scattering]{\includegraphics[width =0.495\linewidth]{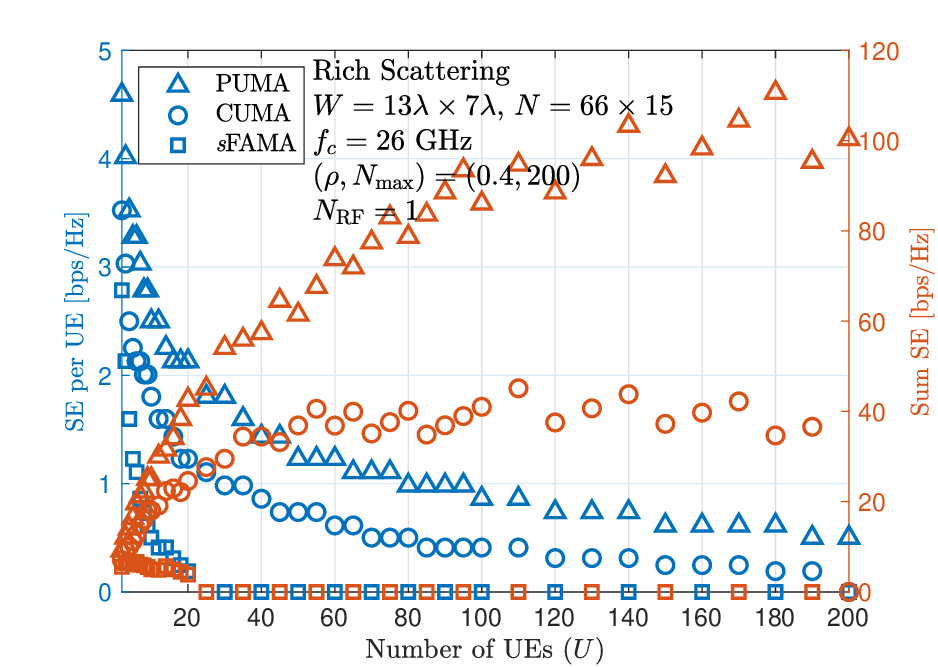}\label{subfig:SEvsU_RS}}
\subfigure[Finite Scattering]{\includegraphics[width =0.495\linewidth]{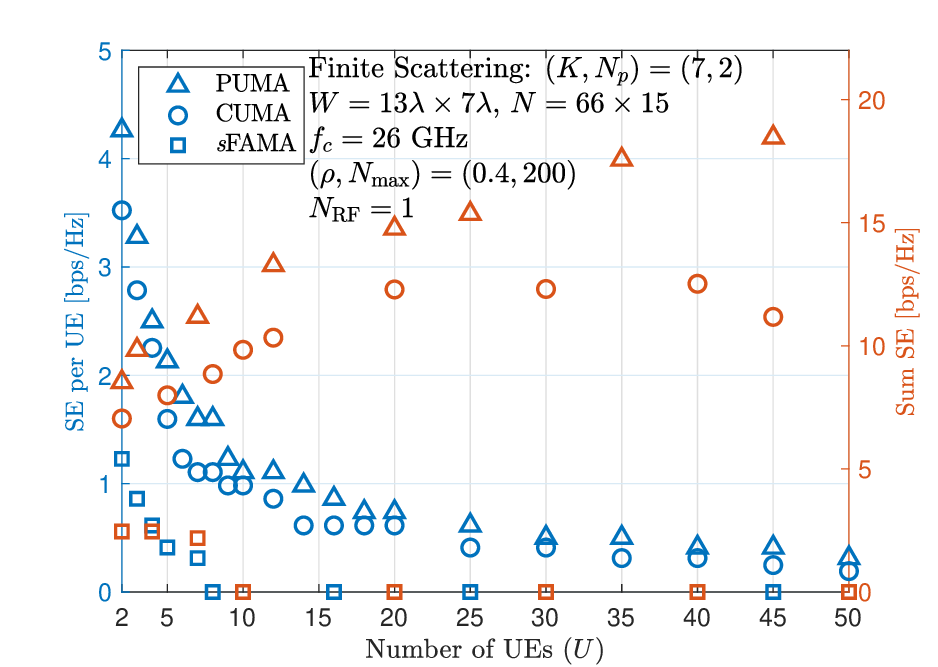}\label{subfig:SEvsU_FS}}
\caption{The spectral efficiency against the number of UEs, $U$, when the link is in sync (${\rm BLER} \leq 0.1$), under (a) rich scattering, and (b) finite scattering with mutual coupling.}\label{Fig:SEvsU}
\end{figure*}

Fig.~\ref{Fig:BLERvsU} depicts the block error rate (BLER) results in relation to the number of UEs, $U$, at a fixed MCS 4. MCS 4 utilizes QPSK and low-density parity-check (LDPC) channel coding with target code rate of $308/1024$, which corresponds to a spectral efficiency (SE) of approximately $0.5~{\rm bit/s/Hz}$ per UE. As can be seen in Fig.~\ref{subfig:BLERvsU_RS} under rich scattering, PUMA exhibits the best performance. It can support a greater number of UEs than both CUMA and \emph{s}FAMA with a fixed MCS, at a certain BLER threshold. Specifically, at an out-of-sync BLER of ${\rm BLER}_{\rm out} = 0.1$, PUMA can support about $220$ UEs with $N_{\rm RF} = 1$, each with an SE of $0.5~{\rm bit/s/Hz}$. In contrast, CUMA supports about $85$ UEs, while \emph{s}FAMA supports only $11$ UEs. Increasing $N_{\rm RF}$ to $2$ results in significant performance gains for CUMA, and PUMA's performance also improves slightly. With $N_{\rm RF} = 2$, PUMA supports about $240$ UEs, CUMA supports $170$ UEs, and \emph{s}FAMA supports $25$ UEs. In the environments characterized by finite scattering, as shown in Fig. \ref{subfig:BLERvsU_FS}, performance declines. In this setting, PUMA surpasses both CUMA and \emph{S}FAMA when $N_{\rm RF} = 1$, supporting roughly $35$ UEs with ${\rm BLER} \geq 0.1$, compared to $23$ UEs for CUMA and $4$ UEs for \emph{s}FAMA. When $N_{\rm RF}$ is increaseed to $2$, PUMA's performance aligns closely with that of CUMA. The performance of \emph{s}FAMA remains relatively low in both cases, exhibiting more pronounced downfall, i.e., a rapid BLER drops as $U$ increases. This phenomenon explains observed decline in data rate for \emph{s}FAMA with increasing $U$ in Fig.~\ref{Fig:RvsUFS_DiffNRF}.

Fig.~\ref{Fig:BLERvsSE} shows the BLER in relation to the SE per UE with a fixed $U=10$. Fig.~\ref{subfig:BLERvsSE_RS} pertains to a rich scattering scenario, and Fig.~\ref{subfig:BLERvsSE_FS} illustrate a finite scattering environment. When $N_{\rm RF} = 1$, PUMA achieves the highest SE at the same BLER threshold, compared to CUMA and \emph{s}FAMA. The per-UE SE gain of PUMA over CUMA is approximately $0.8~{\rm bps/Hz}$ under rich scattering and $0.2~{\rm bps/Hz}$ under finite scattering. The gain is notably higher compared to \emph{s}FAMA. Conversely, when $N_{\rm RF} = 2$, the performance gain of PUMA over CUMA becomes marginal. Under finite scattering, PUMA performs similarly to CUMA when $N_{\rm RF} = 2$, indicating that PUMA is more suitable for systems with limited RF chains, particularly when only one RF chain is employed at each UE.

Using the threshold of out-of-sync BLER of ${\rm BLER}_{\rm out} = 0.1$, we evaluate the maximum SE that the network can support across varying numbers of UEs, $U$. The results are depicted in Fig.~\ref{Fig:SEvsU}, where Fig.~\ref{subfig:SEvsU_RS} corresponds to a rich scattering environment and Fig.~\ref{subfig:SEvsU_FS} illustrates a finite scattering scenario with mutual coupling. The blue curves represent the SE per UE in relation to $U$, while the red curves depict the total sum SE. It is observed that PUMA achieves higher SE under both finite and rich scattering conditions. The gain of PUMA over CUMA is especially pronounced under rich scattering. Also, both PUMA and CUMA outperform \emph{s}-FAMA. The SE of \emph{s}-FAMA declines to $0~{\rm bps/Hz}$ when $U>20$ under rich scattering or $U>8$ in a finite scattering environment, indicating the system is out-of-sync with ${\rm BLER}>0.1$. The trend of the sum SE is consistent with the theoretical data rate from the uncoded simulations. Nonetheless, the SE in Fig.~\ref{subfig:SEvsU_RS} slightly surpasses the data rate in Fig.~\ref{subfig:TheoRvsU_26GHz} when $U$ is small, yet falls below the data rate in Fig.~\ref{subfig:TheoRvsU_26GHz} as $U$ increases. A similar trend is observed when comparing Figs.~\ref{sub@subfig:SEvsU_FS} and \ref{subfig:DiffNRF_26GHzC}. The higher SE at low $U$ can be attributed to the employment of  high-efficiency MCS within this regime, where the SE per UE exceeds $2~{\rm bps/channel~use}$. Conversely, the reduced rate at high $U$ is due to factors, including the overhead associated with pilot signals, cyclic prefix, and guard bands, as well as the more stringent out-of-threshold BLER criteria used. 

\section{Conclusion}\label{sec:conclusion}
In this paper, we proposed a novel FAMA methodology, PUMA, designed to address the extreme connectivity requirements of 6G networks. By advancing the principles of FAMA and CUMA, PUMA incorporates a phased array to facilitate constructive signal aggregation directly at each UE. This methodology exploits the spatial flexibility of FAS to inherently reduce co-user interference {\color{black}without requiring CSI-based precoding at the BS }or SIC at each UE. {\color{black}Our theoretical analysis reveal the underlying scaling behavior under a tractable rich-scattering model, while extensive simulations demonstrate robustness under more practical finite-scattering and coupling conditions.} PUMA markedly improves the achievable data rate relative to existing FAMA and CUMA frameworks. The system is capable of supporting dozens or even hundreds of UEs on the same channel. Notably, PUMA sustains high spectral efficiency and robust interference mitigation even when utilizing only a single RF chain, making it an exceptionally {\color{black}RF-chain-efficient} solution. Consequently, the proposed framework presents a promising and scalable architecture for future wireless networks, facilitating high-performance communication in environments characterized by high device densities and minimal signaling overhead.

\appendices
\section{Proof of Theorem \ref{Theo:GaussianSqrtX}}\label{app:proofTheoSqrtX}
It is established from the rich scattering channel model that $g_k^{(u,u)}\sim \mathcal{CN}(0,\sigma_g^2)$, and $\cov(g_k^{(u,u)},g_\ell^{(u,u)}) = \sigma_g^2\varrho_{k,\ell}$. Consequently, the channel amplitude $\overline{g}_k \triangleq \lvert g_k^{(u,u)}\rvert$ follows Rayleigh distribution, i.e., $\overline{g}_k \sim {\rm Rayleigh} (\sigma_g/\sqrt{2})$. Its covariance is expressed as \cite{midd1996proc}
\begin{align}\label{Eq:covabsg}
\cov (\overline{g}_k,\overline{g}_\ell) & = \mathbb{E}[\overline{g}_k\overline{g}_\ell] - \mathbb{E}[\overline{g}_k]\mathbb{E}[\overline{g}_\ell]\notag\\ 
&= \frac{\pi}{4}\sigma_g^2 \left[{}_2F_1\left(-\frac{1}{2},-\frac{1}{2};1;\varrho_{k,\ell}^2\right) \!-\! 1\right].
\end{align}

The mean of $\sqrt{X}$ in \eqref{Eq:meanSqrtX} is obtained by
\begin{equation}
\mu_1 \triangleq \mathbb{E}[\sqrt{X}] =\sum_{k=1}^{N}  \mathbb{E}[\overline{g}_k] = \frac{N\sqrt{\pi}}{2}\sigma_g.
\end{equation}
The variance of $\sqrt{X}$ can be derived as
\begin{align}\label{Eq:deriveVarSqrtX}
\sigma_1^2 &\triangleq \var[\sqrt{X}] \nonumber\\ 
&= N \var[\overline{g}_k] + 2 \sum_{\ell = 2}^N \sum_{k=1}^{\ell-1} \cov (\overline{g}_k ,\overline{g}_\ell).
\end{align}
This is the classical result for finding the variance of a sum of dependable random variables. Upon substituting \eqref{Eq:covabsg} into \eqref{Eq:deriveVarSqrtX}, we then obtain \eqref{Eq:varSqrtX}.

{\color{black}To account for the spatial dependence among densely deployed FAS ports, the total aperture $W_1\lambda \times W_2\lambda$ can be partitioned into $B$ disjoint effective spatial regions, $\mathcal{S}_b$, such that $\sqrt{X} = \sum_{b=1}^{B}\sum_{k\in \mathcal{S}_b} \overline{g}_k$. The channel amplitudes $\overline{g}_k$ within each region may be strongly correlated, whereas the dependence between sufficiently separated regions decreases with their physical separation. As the number of such effective spatial regions increases, $\sqrt{X}$ consists of an increasing number of weakly dependent spatial contributions. Under the homogeneous scattering model, the centered channel amplitude field, $\{\overline{g}_k - \mathbb{E}\{\overline{g}_k\} : k=1,\dots,N\}$, has finite $(2+\delta)$-order moments for $\delta>0$. Provided that its spatial dependence satisfies the standard weak-dependence conditions required by the central limit theorem (CLT) for spatial random fields \cite[Corollary 1]{jenish2009}, the normalized sum converges in distribution to a standard Gaussian random variable, i.e.,
\begin{equation}
\frac{\sqrt X-\mu_1}{\sigma_1} \xrightarrow{d} \mathcal N(0,1).
\end{equation}
Consequently, when the number of effective spatial degrees of freedom is sufficiently large, $\sqrt X \overset{\mathrm{approx.}}{\sim} \mathcal N(\mu_1,\sigma_1^2)$, which completes the proof.
}


\section{Proof of Theorem \ref{Theo:GaussianSu}}\label{app:proofTheoSu}
Let $h_k^{(\tilde{u})}\triangleq g_k^{(\tilde{u},u)}$, $\theta \triangleq \angle g_k^{(u,u)}$, $\tilde{g}_k \triangleq e^{-j\theta}$, we write $S_{\tilde{u}}$ as
\begin{equation}
S_{\tilde{u}} = \sum_{k = 1}^N h_k^{(\tilde{u})}\tilde{g}_k = \sum_{k=1}^N \tilde{h}_{k}^{(\tilde{u})},
\end{equation}
in which $\tilde{h}_k^{(\tilde{u})} \triangleq h_k^{(\tilde{u})}\tilde{g}_k$, $h_k^{(\tilde{u})} \sim \mathcal{CN}(0,\sigma_g^2)$, $\theta \sim \mathcal{U}[0,2\pi)$, and $h_k^{(\tilde{u})}$ and $\tilde{g}_k$ are independent of each other. 

We first analyze the random variables $\{\tilde{g}_k\}$. The mean of $\tilde{g}_k$ can be directly determined as $\mathbb{E}[\tilde{g}_k] = \mathbb{E}[e^{-j\theta}]= 0$. Consequently, the variance can be derived as $\var [\tilde{g}_k] = \mathbb{E}[{\lvert \tilde{g}_k\rvert}^2] = 1$, given that $\lvert \tilde{g}_k\rvert \equiv 1$. The covariance is computed by calculating the expectation of the phase difference exponential for correlated complex Gaussian variables \cite{midd1996proc}, i.e., 
\begin{align}\label{Eq:covgdivabsg}
 \cov(\tilde{g}_k, \tilde{g}_\ell) & = \mathbb{E}[\tilde{g}_k\tilde{g}_\ell^*] = \mathbb{E}[e^{j(\theta_\ell - \theta_k)}] \nonumber \\
 & = \frac{\pi}{4}\varrho_{k,\ell} {}_2F_1 \left(\frac{1}{2},\frac{1}{2};2;\varrho_{k,\ell}^2\right).
\end{align}

Since $h_k^{(\tilde{u})}$ and $\tilde{g}_k$ are independent of each other, the mean of $\tilde{h}_k^{(\tilde{u})}$ can be calculated as
\begin{equation}\label{Eq:meanterminSu}
\mathbb{E}[\tilde{h}_k^{(\tilde{u})}] = \mathbb{E}[{h}_k^{(\tilde{u})}] \times \mathbb{E}[\tilde{g}_k] = 0,
\end{equation}
and its variance is derived as
\begin{align}\label{Eq:varterminSu}
\var [\tilde{h}_k^{(\tilde{u})}] &= \var[{h}_k^{(\tilde{u})}] \times \var[\tilde{g}_k] + \var[{h}_k^{(\tilde{u})}]\times {(\mathbb{E}[\tilde{g}_k])}^2  \notag\\ 
&\quad +\var [\tilde{g}_k] \times {(\mathbb{E}[{h}_k^{(\tilde{u})}])}^2 \nonumber \\
&= \sigma_g^2.
\end{align}
The covariance can be obtained by
\begin{align}\label{Eq:covterminSu}
\cov(\tilde{h}_k^{(\tilde{u})},\tilde{h}_\ell^{(\tilde{u})})
&\hspace{.5mm} = \mathbb{E} \left[\tilde{h}_k^{(\tilde{u})}{(\tilde{h}_\ell^{(\tilde{u})})}^*\right] \nonumber \\ 
&\overset{(a)}{=} \mathbb{E} \left[{h}_k^{(\tilde{u})}{({h}_\ell^{(\tilde{u})})}^*\right] \mathbb{E} [\tilde{g}_k \tilde{g}_\ell^*] \nonumber \\
& \overset{(b)}{=} \frac{\pi\sigma_g^2}{4}\varrho_{k,\ell}^2 {}_2F_1 \left(\frac{1}{2},\frac{1}{2};2;\varrho_{k,\ell}^2\right),
\end{align}
where $(a)$ uses the fact that $h_k^{\tilde{u}}$ and $\tilde{g}_k$ are independent, $(b)$ is derived from substituting \eqref{Eq:covgdivabsg} and using $\mathbb{E} [{h}_k^{(\tilde{u})}{({h}_\ell^{(\tilde{u})})}^*] = \sigma_g^2\varrho_{k,\ell}$. Then the variance of $s_{\tilde{u}}$ can be calculated by
\begin{align}\label{Eq:varSu1}
\var [S_{\tilde{u}}] = N\var [\tilde{h}_k^{(\tilde{u})}] + 2 \sum_{\ell = 2}^{N} \sum_{k=1}^{\ell-1} \cov(\tilde{h}_k^{(\tilde{u})},\tilde{h}_\ell^{(\tilde{u})}).
\end{align}
Substituting \eqref{Eq:varterminSu} and \eqref{Eq:covterminSu} into \eqref{Eq:varSu1} yields \eqref{Eq:varSu}. Moreover, we can have $\mathbb{E}[S_{\tilde{u}}] = 0$ because of \eqref{Eq:meanterminSu}. Lastly, the argument presented in Appendix \ref{app:proofTheoSqrtX} can be employed to conclude that the series $S_{\tilde{u}}$ is asymptotically normal.

\section{Proof of Theorem \ref{Theo:pdfZ}}\label{app:proofTheopdfZ}
Given the PDF of $X$ in Corollary \ref{Theo:pdfX} and that of $\tilde{Y}$ in Corollary \ref{Theo:pdfYY}, we derive the PDF of $Z = X/\tilde{Y}$ as 
\begin{align}\label{Eq:proodpffZ}
f_{Z}(z) &= \int_0^\infty \tilde{y} f_X(z\tilde{y}) f_{\tilde{Y}}(\tilde{y})d\tilde{y} \nonumber \\
& = \frac{\mu_1^{\frac{1}{2}}e^{-\frac{\mu_1^2}{2\sigma_1^2}}}{2\sigma_1^2 z^{\frac{1}{4}}\Gamma(U-1)}\nonumber\\
&\quad \times \int_0^\infty \tilde{y}^{U-\frac{5}{4}} e^{-\tilde{y}\left(1+\frac{z}{2\sigma_1^2}\right)} I_{-\frac{1}{2}}\left(\frac{\mu_1}{\sigma_1^2}\sqrt{z\tilde{y}}\right) d\tilde{y},
\end{align}
where $\mu_1$ and $\sigma_1^2$ are specified in \eqref{Eq:meanSqrtX} and \eqref{Eq:varSqrtX}, respectively. We have the integral identity \cite[6.643.2]{intifentity}:
\begin{multline}\label{Eq:intIdentity}
\int_0^\infty x^{\mu - \frac{1}{2}} e^{-\alpha x} I_{2\nu}\left(2b\sqrt{x}\right)dx\\
=\frac{\Gamma(\mu+\nu+\frac{1}{2})}{\Gamma(2\nu+1)}b^{-1}e^{\frac{b^2}{2\alpha}} \alpha^{-\mu} \mathcal{M}_{-\mu,\nu} \left(\frac{b^2}{\alpha}\right).
\end{multline}
Applying \eqref{Eq:intIdentity} by setting $\mu = U-\frac{3}{4}$, $\alpha = 1+\frac{z}{2\sigma_1^2}$, $\nu = -\frac{1}{4}$, and $b = \frac{\mu_1\sqrt{z}}{2\sigma_1^2}$, we are able to compute the integral in \eqref{Eq:proodpffZ}, thereby finally deriving \eqref{Eq:pdfZ}, which completes the proof.

\bibliographystyle{IEEEtran}

\begin{thebibliography}{10}
\providecommand{\url}[1]{#1}
\csname url@samestyle\endcsname
\providecommand{\newblock}{\relax}
\providecommand{\bibinfo}[2]{#2}
\providecommand{\BIBentrySTDinterwordspacing}{\spaceskip=0pt\relax}
\providecommand{\BIBentryALTinterwordstretchfactor}{4}
\providecommand{\BIBentryALTinterwordspacing}{\spaceskip=\fontdimen2\font plus
\BIBentryALTinterwordstretchfactor\fontdimen3\font minus
  \fontdimen4\font\relax}
\providecommand{\BIBforeignlanguage}[2]{{%
\expandafter\ifx\csname l@#1\endcsname\relax
\typeout{** WARNING: IEEEtran.bst: No hyphenation pattern has been}%
\typeout{** loaded for the language `#1'. Using the pattern for}%
\typeout{** the default language instead.}%
\else
\language=\csname l@#1\endcsname
\fi
#2}}
\providecommand{\BIBdecl}{\relax}
\BIBdecl
\bibitem{andrews20246gtakes}
J. G. Andrews, T. E. Humphreys and T. Ji, ``6G takes shape,'' {\em IEEE BITS Inf. Theory Mag.}, vol. 4, no. 1, pp. 2--24, Mar. 2024.
\bibitem{ngo2024ultradense}
H. Q. Ngo, G. Interdonato, E. G. Larsson, G. Caire and J. G. Andrews, ``Ultradense cell-free massive MIMO for 6G: Technical overview and open questions,'' \emph{Proc. IEEE}, vol. 112, no. 7, pp. 805--831, Jul. 2024.
\bibitem{ding2025iot}
N. Ding, X. Ouyang, L. Gao, J. Huang and G. Xing, ``An overview on economic analysis of internet of everything,'' {\em IEEE Commun. Surv. \& Tut.}, vol. 27, no. 6, pp. 3742--3771, Dec. 2025.
\bibitem{nguyen20226ginternet}
D. C. Nguyen {\em et al.}, ``6G internet of things: A comprehensive survey,'' {\em IEEE Internet Things J.}, vol. 9, no. 1, pp. 359--383, Jan. 2022.
\bibitem{silva2025distributed}
M. V. da Silva, E. Eldeeb, M. Shehab, H. Alves and R. D. Souza, ``Distributed learning methodologies for massive machine type communication,'' {\em IEEE Internet Things Mag.}, vol. 8, no. 1, pp. 102--108, Jan. 2025.
\bibitem{wang2024extremely}
Z. Wang {\em et al.}, ``Extremely large-scale MIMO: Fundamentals, challenges, solutions, and future directions,'' \emph{IEEE Wireless Commun.}, vol. 31, no. 3, pp. 117--124, Jun. 2024.
\bibitem{clercks2024multiple}
B. Clerckx {\em et al.}, ``Multiple access techniques for intelligent and multifunctional 6G: Tutorial, survey, and outlook,'' \emph{Proc. IEEE}, vol. 112, no. 7, pp. 832--879, Jul. 2024.
\bibitem{ahmed2024unveil}
A. Ahmed {\em et al.}, ``Unveiling the potential of NOMA: A journey to next generation multiple access,'' \emph{IEEE Commun. Surv. \& Tut.}, vol. 27, no. 5, pp. 3099--3164, Oct. 2025.
%
\bibitem{wong2022FAMA}
K. K. Wong and K. F. Tong, ``Fluid antenna multiple access,'' \emph{IEEE Trans. Wireless Commun.}, vol.~21, no.~7, pp. 4801--4815, Jul. 2022.
\bibitem{hong2026famasurvey}
H. Hong {\em et al.}, ``Fluid antenna multiple access for 6G: A holistic review,''  {\em IEEE Open J. Commun. Soc.}, vol. 7, pp. 2607--2633, 2026.
\bibitem{wong2020FAS}
K. K. Wong, K. F. Tong, Y. Zhang, and Z. Zheng, ``Fluid antenna system for {6G}: When {Bruce Lee} inspires wireless communications,'' {\em Elect. Lett.}, vol.~56, no.~24, pp.~1288--1290, Nov. 2020.
\bibitem{wong2021FAS}
K. K. Wong, A. Shojaeifard, K. F. Tong, and Y. Zhang, ``Fluid antenna systems,'' {\em IEEE Trans. Wireless Commun.}, vol. 20, no. 3, pp. 1950--1962, Mar. 2021.
\bibitem{New2024aTutorial}
W. K. New {\em et al.}, ``A tutorial on fluid antenna system for 6G networks: Encompassing communication theory, optimization methods and hardware designs,'' \emph{IEEE Commun. Surv. \& Tut.}, vol. 27, no. 4, pp. 2325--2377, Aug. 2025.
\bibitem{hong2025contemporary}
H. Hong {\em et al.}, ``A contemporary survey on fluid antenna systems: Fundamentals and networking perspectives,'' {\em  IEEE Trans. Netw. Sci. Eng.}, vol. 13, pp. 2305--2328, 2026.
\bibitem{New2026jsac}
W. K. New {\em et al.}, ``Fluid antenna systems: Redefining reconfigurable wireless communications,'' {\em IEEE J. Sel. Areas Commun.},  vol. 44, pp. 1013--1044, 2026.
\bibitem{shen2024design}
Y.~Shen \emph{et al.}, ``Design and experimental validation of mmWave surface wave enabled fluid antennas for future wireless communications,'' \emph{IEEE Antennas \& Wireless Propag. Lett.}, vol. 25, no. 4, pp. 1467--1471, Apr. 2026.
\bibitem{Shamim-2025}
R. Wang {\em et al.}, ``Electromagnetically reconfigurable fluid antenna system for wireless communications: Design, modeling, algorithm, fabrication, and experiment,'' {\em IEEE J. Sel. Areas Commun.}, vol. 44, pp. 1464--1479, 2026.
\bibitem{Liu2025programmable}
B. Liu, K.-F. Tong, K. K. Wong, C.-B. Chae, and H. Wong, ``Programmable meta-fluid antenna for spatial multiplexing in fast fluctuating radio channels,'' {\em Optics Express}, vol. 33, no. 13, pp. 28898--28915, 2025.
\bibitem{Zhangjsac2026}
S. Zhang {\em et al.}, ``Fluid antenna systems enabled by reconfigurable holographic surfaces: Beamforming design and experimental validation,'' {\em IEEE J. Sel. Areas Commun.},  vol. 44, pp. 1417--1431, 2026.
\bibitem{zhang2024pixel}
J. Zhang {\em et al.}, ``A novel pixel-based reconfigurable antenna applied in fluid antenna systems with high switching speed,'' {\em IEEE Open J. Antennas \& Propag.}, vol. 6, no. 1, pp. 212--228, Feb. 2025.
\bibitem{liu2025iot}
B. Liu, T. Wu, K. K. Wong, H. Wong, and K. F. Tong, ``Wideband pixel-based fluid antenna system: An antenna design for smart city,'' {\em IEEE Internet of Things J.}, vol. 13, no. 4, pp. 6850--6862, Feb. 2026.
\bibitem{espinosa2024anew} 
P. Ram\'{i}rez-Espinosa, D. Morales-Jimenez and K. K. Wong, ``A new spatial block-correlation model for fluid antenna systems,'' \emph{IEEE Trans. Wireless Commun.}, vol. 23, no. 11, pp. 15829--15843, Nov. 2024.
\bibitem{zhu2025fluid}
X. Zhu {\em et al.}, ``Fluid antenna systems: A geometric approach to error probability and fundamental limits,'' {\em IEEE Trans. Wireless Commun.}, vol. 25, pp. 17195--17209, 2026
\bibitem{Khammassi2023}
M. Khammassi, A. Kammoun and M.-S. Alouini, ``A new analytical approximation of the fluid antenna system channel,'' {\em IEEE Trans. Wireless Commun.}, vol. 22, no. 12, pp. 8843--8858, Dec. 2023.
\bibitem{zhang2025finite}
Z. Zhang \emph{et al.}, ``Finite-blocklength fluid antenna systems," {\em IEEE Trans. Wireless Commun.}, vol. 25, pp. 22148--22164, 2026.
\bibitem{xu2023channel} 
H. Xu {\em et al.}, ``Channel estimation for {FAS}-assisted multiuser {mmWave} systems,'' {\em IEEE Commun. Lett.}, vol.~28, no.~3, pp.~632--636, Mar. 2024.
\bibitem{zhang2025successive}
Z. Zhang, J. Zhu, L. Dai and R. W. Heath, Jr., ``Successive Bayesian reconstructor for channel estimation in fluid antenna systems,'' {\em IEEE Trans. Wireless Commun.}, vol. 24, no. 3, pp. 1992--2006, Mar. 2025.
\bibitem{ghadi2024on} 
F. Rostami Ghadi {\em et al.}, ``On performance of RIS-aided fluid antenna systems,'' {\em IEEE Wireless Commun. Lett.}, vol. 13, no. 8, pp. 2175--2179, Aug. 2024.
\bibitem{xiao2025fluid}
H. Xiao {\em et al.}, ``Fluid reconfigurable intelligent surfaces: Joint on-off selection and beamforming with discrete phase shifts,'' {\em IEEE Wireless Commun. Lett.}, vol. 14, no. 10, pp. 3124--3128, Oct. 2025.
\bibitem{zhou2024sensing} 
L. Zhou, J. Yao, M. Jin, T. Wu and K. -K. Wong, ``Fluid antenna-assisted ISAC systems,'' {\em IEEE Wireless Commun. Lett.}, vol. 13, no. 12, pp. 3533--3537, Dec. 2024.
\bibitem{zhang2026jsac}
Z. Zhang, \emph{et al.}, ``On fundamental limits for fluid antenna-assisted integrated sensing and communications for unsourced random access," \emph{IEEE J. Sel. Areas Commun.}, vol. 44, pp. 136--149, 2026. 
\bibitem{hong2026FDSIC} 
H. Hong {\em et al.}, ``Fluid antenna system-assisted self-interference cancellation for in-band full duplex communications,'' {\em IEEE Trans. Wireless Commun.}, vol. 25, pp. 7476--7489, 2026.
\bibitem{tang2026FD}
B. Tang {\em et al.}, ``Full-duplex FAS-assisted base station for ISAC,'' {\em IEEE Trans. Wireless Commun.}, vol. 25, pp. 2922--2938, 2026.
\bibitem{hong2025fasofdm}
H. Hong {\em et al.}, ``FAS meets OFDM: Enabling wideband 5G NR,'' {\em IEEE Trans. Commun.}, vol. 73, no. 11, pp. 12884--12898, Nov. 2025.
\bibitem{wong2022fast}
K. K. Wong, K. F. Tong, Y. Chen, and Y. Zhang, ``Fast fluid antenna multiple access enabling massive connectivity,'' \emph{IEEE Commun. Lett.}, vol.~27, no.~2, pp. 711--715, Feb. 2023.
\bibitem{wong2023sFAMA}
K. K. Wong, D. Morales-Jimenez, K. F. Tong, and C. B. Chae, ``Slow fluid antenna multiple access,'' \emph{IEEE Trans. Commun.}, vol.~71, no.~5, pp. 2831--2846, May 2023.
\bibitem{Xu2024revisiting}
H.~Xu {\em et al.}, ``Revisiting outage probability analysis for two-user fluid antenna multiple access system,'' \emph{IEEE Trans. Wireless Commun.}, vol. 23, no. 8, pp. 9534--9548, Aug. 2024.
\bibitem{coma2024slow}
J. P. Gonz\'{a}lez-Coma and F. J. L\'{o}pez-Mart\'{i}nez, ``Slow fluid antenna multiple access with multiport receivers,'' {\em IEEE Wireless Commun. Lett.}, vol. 15, pp. 1280--1284, 2026.
\bibitem{zhang2025sFAMA}
Z. Zhang, \emph{et al.}, ``On fundamental limits of slow-fluid antenna multiple access for unsourced random access,'' \emph{IEEE Wireless Commun. Lett.}, vol. 14, no. 11, pp. 3455--3459, Nov. 2025. 
\bibitem{hong2025multiport}
H. Hong {\em et al.}, ``Multi-port selection for FAMA: Massive connectivity with fewer RF chains than users,'' {\em arXiv preprint}, \url{arXiv:2511.17897}, Nov. 2025.
\bibitem{hong20245gcoded}
H. Hong, K. K. Wong, K. F. Wong, H. Xu and H. Li, ``5G-coded fluid antenna multiple access over block fading channels,'' {\em Elect. Lett.}, vol. 61, no. 1, Jan. 2025.
\bibitem{hong2025coded}
H. Hong, K. K. Wong, K. F. Tong, H. Shin, and Y. Zhang, ``Coded fluid antenna multiple access over fast fading channels,'' \emph{IEEE Wireless Commun. Lett.}, vol.~14, no.~4, pp.~1249--1253, Apr. 2025.
\bibitem{waqar2025turbocharging}
N. Waqar, K. K. Wong, C.-B. Chae and R. Murch, ``Turbocharging fluid antenna multiple access,'' {\em IEEE Trans. Wireless Commun.}, vol. 25, pp. 4038--4052, 2026.
\bibitem{hong2025Downlink}
H. Hong {\em et al.}, ``Downlink OFDM-FAMA in 5G-NR systems,'' {\em IEEE Trans. Wireless Commun.}, vol. 24, no. 12, pp. 10116--10132, Dec. 2025.
\bibitem{Wong2024cuma}
K. K. Wong, C. B. Chae, and K. F. Tong, ``Compact ultra massive antenna array: A simple open-loop massive connectivity scheme,'' {\em IEEE Trans. Wireless Commun.}, vol. 23, no. 6, pp. 6279--6294, Jun. 2024.
\bibitem{wong2024compact}
K. K. Wong, ``Compact ultra massive array (CUMA) with 4 RF chains for massive connectivity,'' in {\em Proc. IEEE Int. Workshop Sig. Proc. Adv. Wireless Commun. (SPAWC)}, pp. 286--290, 10-13 Sept. 2024, Lucca, Italy.
\bibitem{rao2025geometric}
C. Rao {\em et al.}, ``Geometric port selection in CUMA Systems,'' {\em arXiv preprint}, \url{arXiv:2509.20299}, Sept. 2025.
\bibitem{uchino2025PS}
T. Uchino {\em et al.}, ``A compact D-band phase shifter with 0.1-degree phase resolution and ultra-low phase error in 65-nm CMOS,'' {\em IEEE J. Solid-State Circuits}, vol. 60, no. 10, pp. 3566--3576, Oct. 2025.
\bibitem{morishita2025_150}
Y. Morishita {\em et al.}, ``150 GHz-band compact phased-array AiP module for XR applications toward 6G,'' {\em IEEE Microw. Wireless Technol. Lett.}, vol. 35, no. 6, pp. 920--923, June 2025.
\bibitem{sohrabi2016HBF}
{\color{black}F. Sohrabi and W. Yu, ``Hybrid digital and analog beamforming design for large-scale antenna arrays,'' {\em IEEE J. Sel. Topics Signal Proces.}, vol. 10, no. 3, pp. 501--513, Apr. 2016.}
\bibitem{gao2024hbf}
{\color{black}H. Gao, D. Liu and Z. Zhang, ``Energy efficiency maximization for partially-connected hybrid beamforming architecture with low-resolution DACs,'' {\em IEEE Trans. Commun.}, vol. 72, no. 9, pp. 5765--5780, Sept. 2024.}


\bibitem{buzzi2016clustered}
S.~Buzzi and C.~D'Andrea, ``On clustered statistical MIMO millimeter wave channel simulation,'' \emph{arXiv preprint}, \url{arXiv:1604.00648v2}, May 2016.
\bibitem{Clerckx2007MC}
B. Clerckx, C. Craeye, D. Vanhoenacker-Janvier, and C. Oestges, ``Impact of antenna coupling on $2 \times 2$ MIMO communications,'' {\em IEEE Trans. Veh. Technol.}, vol. 56, no. 3, pp. 1009--1018, May 2007.
\bibitem{lee1988analysis}
K. M. Lee and R.-S. Chu, ``Analysis of mutual coupling between a finite phased array of dipoles and its feed network,'' {\em IEEE Trans. Antennas Propag.}, vol. 36, no. 12, pp. 1681--1699, Dec. 1988.
\bibitem{Singh2013MC}
H. Singh, H. L. Sneha, and R. M. Jha, ``Mutual coupling in phased arrays: A review,'' {\em Int. J. Antennas Propag.}, vol. 2013, no. 1, pp. 348123, Apr. 2013.
\bibitem{khan2014anew}
V. Khandelwal and Karmeshu, ``A new approximation for average symbol error probability over log-normal channels,'' {\em IEEE Wireless Commun. Lett.}, vol. 3, no. 1, pp. 58--61, Feb. 2014.
\bibitem{38214} 
``{NR}; {P}hysical layer procedures for data,'' Available [Online]: \url{https://www.3gpp.org/ftp/Specs/archive/38_series/38.214/38214-i40.zip}, Last Accessed on 2024-09-23.
\bibitem{midd1996proc}
David Middleton, ``Processes Derived from the Normal,'' {\em An Introduction to Statistical Communication Theory: An IEEE Press Classic Reissue}, IEEE, 1996, pp.396-437.
\bibitem{jenish2009}
{\color{black}N. Jenish and I. R. Prucha, ``Central limit theorems and uniform laws of large numbers for arrays of random fields,'' {\em Journal of Econometrics}, vol. 150, no. 1, pp. 86--98, 2009.}
\bibitem{intifentity}
D. Zwillinger, V. Moll, I. S. Gradshteyn, and I. M. Ryzhik, ``6-7 Definite integrals of special functions,'' {\em Table of Integrals, Series, and Products (Eighth Edition)}, Academic Press, 2014, pp. 637--775.
\end{thebibliography}

\end{document}